\documentclass[12pt]{article}

\usepackage{amssymb, amsmath, amsthm}
\usepackage{graphicx}
\usepackage{subcaption}
\usepackage{cite}
\usepackage{enumitem}

\newcommand{\cA}{\mathcal{A}}

\newcommand{\ca}{\mathcal{A}}
\newcommand{\cb}{\mathcal{B}}

\newcommand{\cf}{\mathcal{F}}
\newcommand{\scclp}{\sigma}
\newcommand{\rcclp}{r}
\newcommand{\kpot}{\mathcal{K}}
\newcommand{\at}{\tilde{a}}
\newcommand{\bt}{\tilde{b}}
\newcommand{\ct}{\tilde{c}}
\newcommand{\ddt}{\tilde{\delta}}

\newtheorem{theorem}{Theorem}
\newtheorem{lemma}{Lemma}
\newtheorem{definition}{Definition}
\newtheorem{cor}{Corollary}
\newtheorem{prop}{Proposition}

\def\be{\begin{equation}}
\def\ee{\end{equation}}
\def\bea{\begin{eqnarray}}
\def\eea{\end{eqnarray}}

\newcommand{\td}{\text{d}}
\usepackage{hyperref}
\usepackage[left=2cm,right=2cm,top=2cm,bottom=2cm]{geometry}

\title{ \bf{On the uniqueness of supersymmetric AdS$_5$ black holes with toric symmetry}}

\author{James Lucietti\footnote{j.lucietti@ed.ac.uk}\;, Praxitelis Ntokos\footnote{Praxitelis.Ntokos@ed.ac.uk} and Sergei G. Ovchinnikov\footnote{s.g.ovchinnikov@sms.ed.ac.uk} 
\\ \\ \small \sl School of Mathematics and Maxwell Institute for Mathematical Sciences, \\ \small \sl    University of Edinburgh, King's Buildings, Edinburgh, EH9 3JZ, UK }

\date{}

\begin{document}

\maketitle

\begin{picture}(0,0)(0,0)
\put(350, 240){}
\put(350, 225){}
\end{picture}

\begin{abstract}  We consider the classification of supersymmetric AdS$_5$ black hole solutions to minimal gauged supergravity that admit a torus symmetry.  This problem reduces to finding a class of toric K\"ahler metrics on the base space, which in symplectic coordinates are determined by a symplectic potential. We derive the general form of the symplectic potential near any component of the horizon or axis of symmetry, which determines its singular part for any black hole solution in this class,  including  possible new solutions such as black lenses  and multi-black holes.  We find that  the most general known black hole solution in this context, found by Chong, Cvetic, L\"u and Pope (CCLP), is described by a remarkably simple symplectic potential.  We prove that any supersymmetric and toric solution  that is timelike outside a smooth horizon,  with a K\"ahler base metric of Calabi type, must be the CCLP black hole solution or its near-horizon geometry.

\end{abstract}
\vskip1cm

\newpage

\tableofcontents

\newpage

\section{Introduction}

The classification of black holes in anti de Sitter spacetime (AdS) is an open problem of central importance in holography. In this context black holes are dual to states in a CFT  that should  account for the Bekenstein-Hawking entropy~\cite{Witten:1998zw}.  Recent breakthroughs starting with~\cite{Hosseini:2017mds, Cabo-Bizet:2018ehj, Choi:2018hmj, Benini:2018ywd, Zaffaroni:2019dhb} have led to a holographic derivation of the entropy of the known supersymmetric black holes in AdS$_5$~\cite{Gutowski:2004ez, Gutowski:2004yv, Chong:2005hr, Kunduri:2006ek}. Supersymmetry is crucial for performing the required CFT calculations.  Of course, a full understanding of this problem ultimately requires the complete classification of black holes in this context. Surprisingly, even the four-dimensional black hole uniqueness theorem for the Kerr black hole has not been generalised to include a cosmological constant, furthermore, it is expected that it may even fail due to super-radiant instabilities present in AdS spacetimes~\cite{Kunduri:2006qa, Cardoso:2006wa, Dias:2015rxy}. Despite these complications, one may expect that the classification of {\it supersymmetric} AdS black holes is more tractable.   

 In this paper we consider the classification of supersymmetric black hole solutions to five-dimensional minimal gauged supergravity. The most general known asymptotically AdS$_5$ black hole solution in this theory was found by Chong, Cvetic, L\"u and Pope (CCLP) and is a 4-parameter family of topologically $S^3$ black holes specified by its mass $M$, electric charge $Q$ and two angular momenta $J_1, J_2$~\cite{Chong:2005hr}. This family contains a 2-parameter supersymmetric black hole solution which is specified by two independent angular momenta $J_1, J_2$ and charge $Q$ subject to a complicated non-linear constraint  (the mass is fixed by the BPS condition).   An obvious question is whether there are other asymptotically  AdS$_5$ supersymmetric black holes in this theory.  In fact, a complete classification of the possible near-horizon geometries of such black holes has been known for some time~\cite{Kunduri:2006uh, Grover:2013hja}. This reveals that, in fact, the most general smooth near-horizon geometry with compact cross-sections, is locally isometric to that of the supersymmetric CCLP black hole. In particular, this implies that regular black rings do not exist in this theory.\footnote{It is worth noting that supersymmetric near-horizon geometries with ring topology do exist in STU gauged supergravity so black rings can't be ruled out in this more general theory~\cite{Kunduri:2007qy}.}  
 
 However, as emphasised in~\cite{Lucietti:2021bbh}, this still leaves open the possibility of black holes with lens space topology (black lens), black hole spacetimes with nontrivial 2-cycles (bubbles) or even multi-black holes, which are all known in asymptotically flat space~\cite{Kunduri:2014iga, Kunduri:2014kja, Tomizawa:2016kjh, Horowitz:2017fyg, Breunholder:2018roc}. Indeed, the classification of asymptotically flat supersymmetric black hole solutions to minimal (ungauged) supergravity is much better understood.  A complete classification of such solutions has been obtained under the assumption of a biaxial $U(1)^2$-symmetry~\cite{Breunholder:2017ubu} and even for a single axial $U(1)$-symmetry~\cite{Katona:2022bjp} that commutes with supersymmetry.\footnote{Static supersymmetric solutions are given by the Majumdar-Papapetrou multi-black holes~\cite{Lucietti:2020ryg}.} This reveals a very rich moduli space of black hole spacetimes with $S^3, S^2\times S^1$ and lens space $L(p,1)$ horizon topology and an exterior region containing 2-cycles. Strikingly, some of these solutions can have the same conserved charges as the BMPV black hole and greater entropy, raising a puzzle for the original microscopic derivation of black hole entropy in string theory~\cite{Horowitz:2017fyg, Breunholder:2018roc}.   It is therefore important to consider if AdS counterparts of these solutions exist.  As emphasised above, although black rings are ruled out in minimal gauged supergravity, there are no known  obstructions to the other topologies.

Unfortunately, even the classification of supersymmetric backgrounds in gauged supergravity is complicated.  In five-dimensional minimal supergravity, the classification of timelike supersymmetric solutions reduces to finding a K\"ahler metric  that satisfies a complicated 4th order nonlinear PDE for its curvature~\cite{Gauntlett:2003fk, Cassani:2015upa}. This prevents a local classification of solutions. A natural strategy is to seek further simplifying symmetry assumptions that are compatible with the problem at hand.   Asymptotically globally AdS$_5$ spacetimes have an $SO(4)$ rotational symmetry at infinity and therefore one may consider solutions invariant under a subgroup of this symmetry.  In a previous paper, two of the authors considered supersymmetric solutions with $SU(2)$ symmetry and showed that the Gutowski-Reall black hole, or its near-horizon geometry, is the only solution with an analytic horizon~\cite{Lucietti:2021bbh}.  

In the present paper we will consider the classification of supersymmetric solutions to five-dimensional minimal gauged supergravity that are invariant under a toric $U(1)^2$-symmetry (the maximal abelian subgroup of $SO(4)$). This is a particularly notable class since it contains the  CCLP black hole and, in fact, all explicitly known AdS black hole solutions in five-dimensions. This is the analogue of the  aforementioned classification of supersymmetric black holes in ungauged supergravity which was performed under the assumption of such a torus symmetry~\cite{Breunholder:2017ubu} . Unfortunately,  we find that the classification of supersymmetric backgrounds with a torus symmetry in gauged supergravity is a much more complicated problem: it reduces to a problem in toric K\"ahler geometry which we are unable to solve. Nevertheless, we will show that there exists a geometrically defined coordinate system which can describe any black hole solution in this class in a unified manner.  This corresponds to using symplectic coordinates (or `action-angle' coordinates) for the toric K\"ahler base, which have proven to be convenient  in the context of toric K\"ahler geometry~\cite{Abreu}.  In particular, such geometries are described by a symplectic potential which fully encodes the K\"ahler metric (it is related to the K\"ahler potential).  

It is well known that compact toric K\"ahler manifolds are characterised by  Delzant polytopes. The polytope is the image of the moment map defined by the toric symmetry. It corresponds to the orbit space under the toric symmetry and its edges correspond to where the symmetry degenerates or has fixed points. Guillemin gave an explicit formula for the canonical K\"ahler metric associated to a Delzant polytope purely in terms of the combinatorial data that defines the polytope~\cite{Guillemin}. In symplectic coordinates, it corresponds  to the symplectic potential $g= \tfrac{1}{2}\sum_A \ell_A \log \ell_A$ where $\ell_A=0$ are lines  that define the edges of the polytope~\cite{Abreu}.  Curiously, we find that the CCLP black hole is described by a remarkably simple symplectic potential that takes this canonical form, although not every line corresponds to an edge of the orbit space in this case. 

A particularly convenient feature of symplectic coordinates is that they are naturally adapted to describing the axes of symmetry and provide a natural description of the orbit space. Indeed, in symplectic coordinates the components of the axes correspond to lines and horizons to points. They are therefore analogues of the Weyl-Papapetrou coordinates for asymptotically flat electro-vacuum solutions in four and five-dimensions.  Furthermore,  by imposing smoothness at the axes of symmetry and employing the classification of smooth  supersymmetric near-horizon geometries with torus symmetry~\cite{Kunduri:2006uh},  we are able to determine the general form of the symplectic potential near any component of the axes or any horizon (see Theorem \ref{thm:gensymp} for a precise statement of this result). It turns out this fixes the singular part of the symplectic potential and is analogous to Abreu's result for compact toric K\"ahler manifolds~\cite{Abreu}.  In particular, this allows us to write down the singular part of the symplectic potential for possible new solutions such as black lenses and multi-black holes. The existence of such topologically nontrivial black hole solutions thus reduces to determining the smooth part of the symplectic potential. Unfortunately, supersymmetry dictates that the symplectic potential must satisfy a complicated non-linear 8th order PDE so this existence problem is presently out of reach.

Due to the complexity of the supersymmetry constraints it is natural to seek for extra symmetry structures that render the classification problem tractable. Curiously, it turns out that  the supersymmetric CCLP black hole has a K\"ahler base of Calabi type~\cite{Cassani:2015upa}.  K\"ahler surfaces of Calabi type naturally appear in the classification of K\"ahler surfaces that admit a hamiltonian 2-form~\cite{Apostolov2001TheGO}. We will give a self-contained definition of Calabi type in the context of toric K\"ahler surfaces in terms of a certain orthogonality property of the associated moment maps (analogous to orthotoric K\"ahler, see Definition \ref{def}). For K\"ahler base metrics of this type we are in fact able to obtain a complete classification, which is the main result of this paper.

\begin{theorem}\label{main-theorem}
Any supersymmetric toric solution to five-dimensional minimal gauged supergravity that is timelike outside a smooth horizon with compact cross-sections, with a K\"ahler base of Calabi type, is locally isometric to the CCLP black hole or its near-horizon geometry.
\end{theorem}

We emphasise that no global assumptions are required for this result, so in particular it rules out asymptotically locally AdS$_5$ supersymmetric black holes (other than quotients of CCLP).  Furthermore, it also does not assume the horizon is connected, and therefore rules out multi-black holes in this symmetry class. In this sense our result is analogous to the uniqueness theorem obtained for $SU(2)$-invariant supersymmetric solutions~\cite{Lucietti:2021bbh}. 

The proof of Theorem \ref{main-theorem} is roughly as follows.  A toric K\"ahler metric of Calabi type is determined by two functions of single variables and in this case the supersymmetry constraint simplifies, although a general local solution is still unavailable.  However, by comparing to the general form of the near-horizon geometry~\cite{Kunduri:2006uh} one can in fact completely fix one of these functions. The supersymmetry constraint now reduces to an ODE problem which can be completely solved in the generic case.  There is a special case (which corresponds to the solution possessing  an $SU(2)$ symmetry) which leads to a 5th order ODE which was previously encountered and can be solved under the additional assumption of an analytic horizon~\cite{Lucietti:2021bbh}.\footnote{Therefore, strictly speaking, for these non-generic near-horizon geometries with enhanced symmetry, we have proven Theorem \ref{main-theorem} under the stronger assumption that the horizon is analytic.} It then turns out that the general solution is the CCLP black hole or its near-horizon geometry.

The organisation of this paper is as follows.  In Section \ref{sec:toric} we consider supersymmetric solutions to gauged supergravity with toric symmetry, introduce symplectic coordinates and the symplectic potential, derive the general form for the symplectic potential near any horizon or axis of symmetry (see Theorem \ref{thm:gensymp}), present the CCLP black hole in this formalism, and then give the singular part of the potential for the simplest possible topologically nontrivial black hole solutions (lenses and double-black holes).  In Section \ref{sec:Calabi} we restrict to supersymmetric solutions with K\"ahler bases of Calabi type and prove Theorem \ref{main-theorem}. A number of details are relegated to the appendices.

\section{Supersymmetric solutions with toric symmetry}
\label{sec:toric}

\subsection{General local classification}
The bosonic field content of five-dimensional minimal gauged supergravity is a metric $\mathbf{g}$ and a Maxwell field $F$ defined on a five-dimensional spacetime manifold $M$.
In any open region $U\subset M$ where the supersymmetric Killing vector $V$ (a spinor bilinear) is timelike, the general {\it local} form for a supersymmetric solution is~\cite{Gauntlett:2003fk} (we work in the conventions of~\cite{Gutowski:2004ez}), 
\be\label{metricform}
\mathbf{g} = -f^2( \td t+\omega)^2+ f^{-1} h  \; ,
\ee
where $V=\partial_t$,  $h$ is a K\"ahler metric on the base $B$ orthogonal to the orbits of $V$, and $f$ and $\omega$ are a function and 1-form on $B$. The Maxwell field takes the form
\be
F = \frac{\sqrt{3}}{2} \td (f( \td  t+\omega)) - \frac{1}{\sqrt{3}} G^+ - \frac{\sqrt{3}}{\ell f} X^{(1)}  \; ,  \label{maxwell}
\ee
where $G^\pm = \tfrac{1}{2}f ( \td\omega \pm \star_4 \td\omega )$, $\star_4$ is the Hodge star operator with respect to the base metric $h$, $X^{(1)}$ is the K\"ahler form and the orientation on $B$ is such that $X^{(1)}$ is anti-self dual on $B$.  Given a K\"ahler base, $f, G^+$ are completely fixed in terms of its curvatures 
\be
f^{-1} = - \frac{\ell^2}{24} R, \qquad G^+ = -\frac{\ell}{2} (\mathcal{R}- \tfrac{1}{4} X^{(1)} R)  \; ,    \label{fGp}
\ee
where $\mathcal{R}_{ab}= \tfrac{1}{2} R_{ab}^{~~cd} X^{(1)}_{cd}$ is the Ricci form.  It is worth noting that combining (\ref{maxwell}) and (\ref{fGp}) allows us to deduce the gauge field (up to a gauge transformation),
\be
 A=  \frac{\sqrt{3}}{2} f( \td  t+\omega)) + \frac{\ell}{2\sqrt{3}}P  \; ,   \label{gf}
\ee
 where $\mathcal{R}= \td P$, i.e., $P$ is the potential for the Ricci form.

Given such a K\"ahler surface there exists a basis of ASD 2-forms  $X^{(i)}$, $i=1,2,3$, where say $X^{(1)}$ is the K\"ahler form, that satisfy the quaternion algebra
\be
X^{(i)}_{ac} X^{{(j)} c}_{\phantom{(j)}~~b}= -\delta^{ij} h_{ab}+ \epsilon^{ijk} X^{(k)}_{ab}   \label{quat}
\ee
and
\be
\nabla_a X^{(2)}_{bc}= P_a X^{(3)}_{bc}, \qquad \nabla_a X^{(3)}_{bc}= -P_a X^{(2)}_{bc}  \;.  \label{covX23}
\ee
  Using this basis we may expand
\be
G^- = \frac{\ell}{2 R} \lambda_i X^{(i)}  \; ,  \label{Gm}
\ee
where $\lambda_1= \tfrac{1}{2} \nabla^2 R+\tfrac{2}{3} R_{ab}R^{ab} -\tfrac{1}{3}R^2$. The other functions $\lambda_{2,3}$ are determined by the integrability condition for
\be
\td \omega= f^{-1}( G^+ +G^-) \; ,   \label{dom}
\ee
i.e. requiring that the r.h.s is closed. This integrability condition also implies that the K\"ahler base is not free to choose: it must satisfy the following complicated 4th order PDE for its curvature~\cite{Cassani:2015upa}
\be
\nabla^2 \left(\frac{1}{2} \nabla^2 R+\frac{2}{3} R_{ab}R^{ab}- \frac{1}{3} R^2 \right)+ \nabla^a(R_{ab}\nabla^b R)=0 \; .  \label{PDE}
\ee
Conversely, given a K\"ahler metric which satisfies this equation, one can solve for $\omega$ and hence the full spacetime metric and gauge field can be reconstructed as indicated above. Thus the classification of timelike supersymmetric solutions in this theory reduces to that of K\"ahler metrics that obey this PDE.

It is also worth noting that $V$ is only defined up to constant rescalings (since the Killing spinor is).  Those act on the time coordinate and the supersymmetric data as
\begin{equation}\label{time-resc}
t\to Kt\,,\qquad\omega\to K\omega\,,\qquad h\to K^{-1}h\,,\qquad f\to K^{-1}f\,,
\end{equation}
for constant $K\neq 0$. Of course, the five-dimensional metric $\mathbf{g}$ and Maxwell field $F$ are invariant under such rescalings.

\subsection{Supersymmetric solutions with toric symmetry}
We wish to classify supersymmetric solutions $(M, \mathbf{g}, F)$ to five-dimensional minimal gauged supergravity possessing two commuting axial Killing fields. In particular we will say that a supersymmetric solution admits toric symmetry if: 
\begin{enumerate}
\item there is a torus $T\cong U(1)^2$ isometry generated by spacelike Killing fields $m_i$, $i=1,2$, both normalised to have $2\pi$ periodic orbits;  these are defined up to $m_i \to A_{i}^{~j} m_j$ where $A\in GL(2, \mathbb{Z})$;
\item  the supersymmetric Killing $V$ is complete and commutes with the $T$-symmetry, that is $[V, m_i]=0$, so the spacetime isometry group is $\mathbb{R}\times U(1)^2$;
\item 
the Maxwell field is $T$-invariant $\mathcal{L}_{m_i} F=0$; 
\item the axis defined by  $\{p\in M | \det \mathbf{g}(m_i, m_j)|_p=0 \}$
is nonempty.
\end{enumerate}

We will now deduce the constraints imposed by such a toric symmetry for timelike supersymmetric solutions, that is, in any open region $U \subset M$ where $V$ is strictly timelike. For simplicity we will also assume that $U$ is simply connected.

Under these assumptions it follows that the data $(f, h, X^{(1)})$ on the base space $B$ are all invariant under the toric symmetry, and that we can choose a gauge for $\omega$ such that it is also invariant, see~\cite[Lemma 1]{Lucietti:2021bbh}. Recall that these gauge transformations act as $\omega \to \omega +\td\lambda, \; t \to t-\lambda$ where $\lambda$ is a function on $B$, so this amounts to choosing a gauge where $\mathcal{L}_{m_i} t=0$. In particular, the $T$-symmetry is holomorphic,  that is $\mathcal{L}_{m_i} X^{(1)}=0$, which is equivalent to  $\td \iota_{m_i} X^{(1)}=0$ since $X^{(1)}$ is closed.  Now, $X^{(1)}$ is a globally defined 2-form on spacetime (defined as a spinor bilinear) and therefore $\iota_{m_i} X^{(1)}$ is a globally defined 1-form on $U$.  Therefore, we deduce the existence of functions $x_{i}$ (moment maps) on $U$ such that 
\be
\iota_{m_i} X^{(1)} = -\td x_{i}  \; .\label{moment}
\ee
This shows that the $T$-symmetry is Hamiltonian and therefore $(h, X^{(1)}, B)$ is a toric K\"ahler structure.   The moment maps $x_{i}$ define a canonical coordinate system $(x_{i},  \phi^i)$  for any toric K\"ahler structure that is adapted to the toric symmetry such that $m_i= \partial_{\phi^i}$. In terms of these coordinates~\cite{Abreu},
\bea
&&h = G^{ij}(x) \td x_{i} \td x_{j}+ G_{ij}(x) \td\phi^i \td\phi^j, \label{h_toric}  \\
&&G^{ij} = \partial^{i}\partial^{j} g \label{hessiang} \\
&& X^{(1)}= \td x_{i} \wedge \td\phi^i , \label{KahlerForm-gen}
\eea
where $g=g(x)$ is the symplectic potential, $G_{ij}$ is the matrix inverse of the Hessian $G^{ij}$ and we  have introduced the notation $\partial^{i}:=\partial/\partial{x_i}$. Observe that the symplectic potential is only defined up to a linear function of $x_i$. The coordinates $(x_{i}, \phi^i)$ are called symplectic (or Darboux) coordinates because they are adapted to the K\"ahler form (which is also a symplectic form). We give a self-contained derivation of this coordinate system in  Appendix \ref{app:symplcoords}.  These give a coordinate system $(x_i, \phi^i)$ on $B$ and hence on $U$ away from the axis.

We will now consider computing the remaining data $(f, \omega)$. For this we introduce a basis of ASD 2-forms $X^{(i)}$ where $X^{(1)}$ is the K\"ahler form given above. We find that 
\bea
&&X^{(2)}= \sqrt{\det G}\td \phi^1 \wedge \td \phi^2 - \frac{1}{ \sqrt{\det G}} \td x_{1} \wedge \td x_{2} \; , \label{eq:X2} \\
&&X^{(3)}= 2\sqrt{\det G}\,G^{i[1} \td x_{i} \wedge \td \phi^{2]}\,,  \label{eq:X3}
\eea
where $\det G := \det G_{ij}$, 
are ASD and together with $X^{(1)}$ satisfy the quaternion algebra (\ref{quat}).  Using these one can compute the potential $P$ for the Ricci form from (\ref{covX23}) which is
\be
P= P_i \td \phi^i \; , \qquad P_i=- \tfrac{1}{2} G_{ij} \partial^{j} \log \det G = - \tfrac{1}{2} \partial^{j}G_{ij} \; .
\ee
The scalar curvature is
\be
R= - \partial^{i} \partial^{j} G_{ij}  \; ,
\ee
which from (\ref{fGp}) immediately gives
\be
f^{-1}= \frac{\ell^2}{24} \partial^{i} \partial^{j} G_{ij}\; .   \label{ftoric}
\ee
It remains to solve for the 1-form $\omega$.

To this end, observe that given any closed 2-form $\Omega$ on $M$ invariant under the toric symmetry, $\iota_{m_1} \iota_{m_2} \Omega$ must be a constant. Therefore, since by assumption we have a nonempty axis,  we deduce that $\iota_{m_1} \iota_{m_2} \Omega=0$ for any closed 2-form invariant under the toric symmetry.  Thus,  contracting (\ref{maxwell}) we deduce $\iota_{m_1} \iota_ {m_2} G^+=0$ which simplifies to 
\be
\iota_{m_1} \iota_ {m_2} \star_4 \td\omega=0   \iff  m^\flat_1 \wedge m^\flat_2 \wedge \td\omega=0  \; ,
\ee
where $m_i^\flat$ is the dual of $m_i$ with respect to the base metric $h$. Thus, since $m^\flat_i = G_{ij} \td\phi^j$  this  becomes $\td\phi^1 \wedge \td\phi^2 \wedge \td\omega=0$. Writing $\omega = \alpha^i \td x_i+ \omega_i \td\phi^i$ this implies that $\alpha^i =\partial^{i} \alpha$ for some function $\alpha(x)$ on $B$. Therefore, by a gauge transformation $\omega \to \omega+ \td\lambda$ which preserves our gauge $\mathcal{L}_{m_i} \omega=0$ we may set $\alpha=0$. We have thus shown that we can always choose a gauge such that
\be
\omega = \omega_i \td\phi^i  \; ,   \label{om}
\ee
where the components $\omega_i$ are invariant under the toric symmetry and hence are functions $\omega_i=\omega_i(x)$ (observe that we can write $\omega_i = \iota_{{m_i}}\omega$ as invariants).

Now consider the equation for $\omega$ which is given by  (\ref{dom}).  It is easy to see that from the form of $\omega$ derived above (\ref{om}), i.e. the absence of $x_{i}$ components,   $\td \omega$ will not have any $x_1x_2$ components. Therefore the coefficient of $X^{(2)}$ in the expansion (\ref{Gm})  for $G^-$ must vanish,  $\lambda_2=0$.   The coefficient $\lambda_3$ must be invariant under the toric symmetry (since $X^{(3)}$ is invariant) and hence is a function $\lambda_3=\lambda_3(x)$.
The integrability condition for (\ref{dom}) now gives a linear first order PDE for $\lambda_3(x)$.  The integrability condition for this PDE is equivalent to (\ref{PDE}) in the toric class we are considering. Unfortunately, we have not found a useful way to write this integrability condition in the toric class nor a general solution for $\omega_i$.  This prevents us from giving an explicit general local solution for the toric class. 

To summarise, so far we have shown that a timelike supersymmetric toric solution in symplectic coordinates on the base takes the form
\be
\mathbf{g}= -f^2(\td t+\omega_i \td\phi^i)^2 +f^{-1} G_{ij} \td\phi^i \td\phi^j +   f^{-1} G^{ij} \td x_i \td x_j  \; ,  \label{susytoric}
\ee 
where $G^{ij}$ is given by (\ref{hessiang}) and $G_{ij}$ is its matrix inverse,  $f$ is given  by (\ref{ftoric}),  $\omega_i$ is determined by (\ref{dom}) and the symplectic potential $g$ must satisfy (\ref{PDE}). We emphasise that any such solution is determined by a single real function, the symplectic potential $g$, which must satisfy a complicated PDE determined by supersymmetry.  It is useful to record the following spacetime invariants
\be
\mathbf{g}(V,  V)= - f^2, \qquad \mathbf{g}(V, m_i )= - f^2 \omega_i, \qquad \mathbf{g}(m_i , m_j) = f^{-1} G_{ij}- f^2 \omega_i \omega_j  \; . \label{Killinginv}
 \ee
 From these we can  define the Gram matrix of Killings fields  $K_{\alpha\beta}:= \mathbf{g}(K_\alpha , K_\beta)$, where  $\alpha=0, i$, $K_0=V, K_i=m_i$, which gives 
 \be
 \det K_{\alpha\beta}= - \det G_{ij}  \; .   \label{gramK}
 \ee
 Observe that all these spacetime functions are invariant under the Killing fields.
 
 We will be interested in solutions that possess a supersymmetric  horizon, that is, a horizon that is invariant under the supersymmetric Killing field $V$, the most notable example being  an event horizon.  In what follows we will show that  the horizon and the axis have a simple description in symplectic coordinates and the singular behaviour of the symplectic potential near a horizon or the axis can be completely fixed.

 \subsection{Horizons, axis and orbit space}\label{ax-hor-orb}
 
 We are interested in spacetimes containing a black hole region. In this context $M$ will denote the domain of outer communication and the event horizon is its inner boundary. It has been shown  that under certain reasonable global assumptions the orbit space $\hat{M} := M/ (\mathbb{R}\times U(1)^2)$, where the $\mathbb{R}\times U(1)^2$ action is given by the flow of $V, m_i$,  is a 2-dimensional manifold with boundaries and corners~\cite{Hollands:2012xy, Hollands:2008fm}. The global assumptions that are made to obtain this result are that $M = \mathbb{R}\times \Sigma$, where $V$ is tangent to $\mathbb{R}$ and $\Sigma$ is a simply connected manifold \footnote{It is also assumed that the torus action is effective and has no discrete isotropy groups. The latter assumption has been removed for asymptotically flat spacetimes~\cite{Hollands:2008fm}.}.   The boundary segments of $\hat{M}$ correspond to either horizon components, or components of the axis where an integer linear combination of  the axial Killing fields $v:=v^im_i$, where $v^i\in \mathbb{Z}^2$ are coprime, vanishes. The corners correspond to fixed points of the toric symmetry (i.e. where $m_1=m_2=0$) or where the axis meets a horizon.  In order to avoid an orbifold singularity at the fixed points of the torus symmetry one must satisfy~\cite{Hollands:2008fm}
\be
\det ( \mathbf{v}, \mathbf{w})=\pm 1  \; , \label{orbifold}
\ee
where $\mathbf{v}=(v^1, v^2) \in \mathbb{Z}^2$ etc are the components of the axial vectors $v=v^i m_i$ and $w=w^i m_i$ that vanish on the adjacent axis components.

Now, as noted above the spacetime invariants (\ref{Killinginv}) are preserved by the Killing fields and hence descend to functions on the orbit space $\hat{M}$.  Furthermore, the functions $x_i$ defined by (\ref{moment}) are also preserved by the Killing fields and hence can be used as local coordinates on  the orbit space.  In fact, the orbit space inherits a metric wherever $V$ is timelike, defined by $q_{\mu\nu}:= \mathbf{g}_{\mu\nu} - K^{\alpha\beta}\mathbf{g}_{\alpha\mu}\mathbf{g}_{\beta \nu}$ where $K^{\alpha\beta}$ is the inverse matrix of $K_{\alpha\beta}$, which using (\ref{susytoric})
gives
\be
q =   f^{-1} G^{ij} \td x_i \td x_j  \; .
\ee
This gives a Riemannian metric on the orbit space if and only if $\det G_{ij}>0$ which from (\ref{gramK}) can be seen to correspond to the region away from the axis.
 
We will now describe the orbit space in terms of symplectic coordinates. We expect that it is possible to show that symplectic coordinates give a global chart on $\hat{M}$ (at least if $V$ is strictly timelike in the exterior region), although this would presumably require making certain global assumptions and in particular that the spacetime is asymptotically (locally) AdS \footnote{This would be analogous to showing that the Weyl-Papapetrou coordinates give a global chart on the interior of the orbit space for asymptotically flat stationary and (bi)axisymmetric spacetimes, see e.g.~\cite{Hollands:2012xy}.}.  We will not make any such global assumptions in this paper and hence not pursue this question here. Instead, in this subsection, we  will simply assume that a single symplectic chart covers all components of the axis, as is  indeed the case for the known solutions discussed below in section \ref{sec:ex}.
 
 We will first show that the axis has a simple description in terms of symplectic coordinates and one can fully determine the singular behaviour of the symplectic potential near the axis.

\begin{lemma} \label{lem:axis}
Consider a supersymmetric and toric solution that is timelike on a neighbourhood of a component of the axis defined by the vanishing of $v:= v^i m_i$, where $(v^i )\in \mathbb{Z}^2$ are coprime integers,  and let
\be
\ell_v(x):= v^i x_i +c_v  \; , \label{axisline}
\ee
where $c_v$ is a constant. Then:
\begin{enumerate}
\item  The axis component corresponds to a straight line $\ell_v(x)=0$ in symplectic coordinates and away from the axis
\be
\ell_v(x)>  0
\ee
\item The symplectic potential can be written as
\be
g= \frac{1}{2}\ell_v(x) \log \ell_v(x) + \tilde{g}  \; ,
\ee
where $\tilde{g}$ is smooth at $\ell_v(x)=0$.
\end{enumerate}
\end{lemma} 

 \begin{proof}
From  (\ref{moment}) it follows that $\iota_v X^{(1)} = - \td \ell_v$. Since $X^{(1)}$ is a non-degenerate 2-form on the base $B$,  we see that $v = v^i m_i=0$ if and only if $\td \ell_v=0$.  Thus $\ell_v$ is a constant on the axis component defined by $v=0$, and we may choose $c_v$ so that $\ell_v(x)=0$ on the axis component.  We must have $\ell_v(x) \neq 0$ away from the axis component and since this corresponds to a boundary component of the orbit space we can choose $v^i$ such that $\ell_v(x)>0$ corresponds to the interior of $\hat{M}$. 

To prove the second part we will derive the geometry near a component of the axis (\ref{axisline}). By a suitable $GL(2,\mathbb{Z})$ transformation we can always arrange $v=m_1=\partial_{\phi^1}$ and define new symplectic coordinates so that  $\ell_v(x)= x_1$. Thus using (\ref{Killinginv}) and that the inner products  $V\cdot m_1=0$ and $m_1\cdot m_i=0$ on this axis, we deduce that $\omega_1=0$ and $G_{1 i}=0$ at $x=0$ (we set $x=x_1, y=x_2$ for clarity). Since these invariants are smooth functions it follows that $G_{1i}= O(x)$ and $\omega_1= O(x)$ so in particular we can write
\be
G_{ij} \td \phi^i \td \phi^j = x a(x,y) (\td \phi^1)^2 + 2 x b(x,y) \td \phi^1 \td \phi^2 + c(x,y) (\td \phi^2)^2  \; ,
\ee
where $a, b, c$ are smooth at $x=0$ and $c>0$ at $x=0$ to ensure we are not at a fixed point of the toric symmetry (this is because $G_{22}=0$ and (\ref{Killinginv}) implies $m_2 \cdot m_2=0$ and hence $m_2=0$). It follows that  $\det G_{ij} = x \delta$ where $\delta= ac - x b^2$ and
\be
G^{ij} \td x_i \td x_j = \frac{1}{\delta} \left( \frac{c}{x} \td x^2 - 2 b \td x \td y+ a \td y^2 \right)  \; .
\ee 
Now consider  smoothness of the K\"ahler base metric (\ref{h_toric}) near $x=0$. To perform the analysis it is useful to set $x=r^2$, in terms of which the full metric reads,
\be
h =  \frac{4 c}{\delta} \td r^2+ r^2 a (\td \phi^1)^2+  \frac{1}{\delta} \left(- 4  b r \td r \td y+ a \td y^2 \right)  + 2 r^2 b \td \phi^1 \td \phi^2 + c (\td \phi^2)^2 \; .
\ee
It is then standard to verify that this metric extends to a smooth metric at $r=0$ if and only if 
\be
a|_{x=0}=2  \; ,
\ee
where recall we assume the period of $\phi^i$ is $2\pi$ (this can be seen by changing to cartesian coordinates in the $(r, \phi^1)$ plane).  It follows that as $ x\to 0$ 
\be
G^{ij} \td x_i \td x_j = \left( \frac{1}{2x}+O(1) \right)   \td x^2 +O(1) \td x \td y+ (c^{-1} + O(x) )\td y^2  \; .
\ee
Integrating for the symplectic potential gives
\be
g(x,y) = \frac{1}{2} x \log x + \tilde{g}(x,y) \; ,
\ee
where $\tilde{g}$ is smooth at $x=0$ with $\tilde{g}_{yy}= c^{-1}+O(x)$. Therefore, we have fully determined the singular behaviour of the symplectic potential near an axis of symmetry. Returning to our  original basis we  deduce that the symplectic potential takes the claimed form. 
\end{proof}

 Recall that any fixed point must occur at the intersection of two axis components, so Lemma \ref{lem:axis} immediately implies that fixed points of the toric symmetry are single points in symplectic coordinates that occur at the intersection of the two corresponding axis line components.   In fact, the singular behaviour of the symplectic potential at a fixed point arises purely from that of the two components of the axis which define it, as our next result shows.
 
\begin{lemma} \label{lem:fixed-point}
Consider a fixed point defined by the vanishing of  $v:= v^i m_i$ and  $w:= w^i m_i$, where $(v^i )\in \mathbb{Z}^2$ and $(w^i )\in \mathbb{Z}^2$ are pairs of coprime integers that satisfy \eqref{orbifold}. Let $\ell_v(x)=0$ and $\ell_w(x)=0$ denote the lines that correspond to the axis components defined by the vectors $v, w$ as  in (\ref{axisline}).
\begin{enumerate}
\item The fixed point in symplectic coordinates is given by the point defined by $\ell_v(x)=\ell_w(x)=0$.
\item The symplectic potential takes the form
\begin{equation}
g=\frac{1}{2}\ell_{v}(x)\log\ell_{v}(x)+\frac{1}{2}\ell_{w}(x)\log\ell_{w}(x)+\tilde{g}\,,
\end{equation}
where $\tilde{g}$ is smooth at both axis components and at the fixed point.  Furthermore, a neighbourhood of the fixed point in the base is diffeomorphic to $\mathbb{R}^4$.
\end{enumerate}
\end{lemma} 

\begin{proof}
The proof is along the lines of the proof of the second part of Lemma  \ref{lem:axis} applied for each axis component, which also gives us the behaviour of the K\"ahler metric at the intersection point.  Since $v$ and $w$ satisfy \eqref{orbifold}, a $GL(2,\mathbb{Z})$ transformation can be used to set $v=m_1=\partial_{\phi^1}$, $w=m_2=\partial_{\phi^2}$ and define new symplectic coordinates so that $\ell_v(x)= x_1$, $\ell_w(x)= x_2$. Using again the invariants (\ref{Killinginv}) (we set again $x=x_1, y=x_2$ ) together with their smoothness at the axes,  we can now write
\begin{equation}
G_{ij}\td\phi^{i}\td\phi^{j}=x\at(x,y)(\td\phi^{1})^{2}+2xy\bt(x,y)\td\phi^{1}\td\phi^{2}+y\ct(x,y)(\td\phi^{2})^{2}\,,
\end{equation}
where $\at, \bt, \ct$ are smooth on both the lines $x=0$ and $y=0$ and the point $x=y=0$. The K\"ahler metric reads
\begin{equation}\label{h-fixed-point}
h=\frac{4\ct}{\ddt}\td r_{1}^{2}+\at r_{1}^{2}(\td\phi^{1})^{2}+\frac{4\at}{\ddt}\td r_{2}^{2}+\ct r_{2}^{2}(\td\phi^{2})^{2}-\frac{8\bt}{\ddt}r_{1}r_{2} \td r_{1}\td r_{2}+2r_{1}^{2}r_{2}^{2}\bt\td\phi^{1}\td\phi^{2}\,,
\end{equation}
where $\ddt=\at\ct-xy\bt^{2}$ and we have set $x=r_{1}^{2}$, $y=r_{2}^{2}$. The metric is free of conical singularities and smooth at $r_{1}=0$ and $r_{2}=0$ if and only if
\begin{equation}
\at|_{x=0}=\ct|_{y=0}=2\,, 
\end{equation}
and by continuity this also holds at the fixed point. Then, the K\"ahler metric \eqref{h-fixed-point} approaches the flat metric $\sum_{i=1}^2 \td r_i^2+r_i^2 (\td \phi^i)^2$  on $\mathbb{R}^{4}$ as $r_1, r_2\to 0$ with the error terms smooth at $r_i=0$ (this can be seen by converting to cartesian coordinates in the $(r_i, \phi^i)$ planes). This shows that the K\"ahler base near the fixed point is diffeomorphic to $\mathbb{R}^4$.   Furthermore, we have
\begin{equation}
G^{ij}\td x_{i}\td x_{j}=\Big(\frac{1}{2x}+O(1)\Big)\td x^{2}+\Big(\frac{1}{2y}+O(1)\Big)\td y^{2}+O(1)\td x\td y\,,
\end{equation}
where the $O(1)$ terms are smooth at $x=0$, $y=0$ and $x=y=0$. Integrating we find 
\begin{equation}
g(x,y)=\frac{1}{2}x\log x+\frac{1}{2}y\log y+\tilde{g}(x,y)\,,
\end{equation}
where $\tilde{g}$ is smooth at $x=0$, $y=0$ and $x=y=0$, 
thus completing  the proof.
\end{proof}
  
 We now show that a horizon also has a simple description in symplectic coordinates.  The event horizon of a black hole spacetime must be invariant under any Killing field and hence in particular under the supersymmetric Killing field $V$. We will therefore consider supersymmetric horizons, that is, horizons that are invariant under the supersymmetric Killing field $V$. We will assume that each connected component of the horizon possesses a cross-section $S$, that is, a spacelike submanifold everywhere transverse to $V$.  The general form of the metric near a connected component of a supersymmetric horizon can be written in Gaussian null coordinates (GNC) $(v, \lambda, y^a)$ and takes the form~\cite{Gutowski:2004ez} 
\be\label{NH-metric}
\mathbf{g} = -\lambda^2 \Delta^2 \td v^2+ 2 \td v \td \lambda+ 2\lambda h_a \td v \td y^a+ \gamma_{ab}\td y^a \td y^b  \; ,
\ee
where $V= \partial_v$ is the supersymmetric Killing field, $\lambda$ is an affine parameter for null transverse geodesics synchronised so $\lambda=0$ at the horizon, and $y^a$ are coordinates on a cross-section $S$.  Here the data $\Delta, h_a, \gamma_{ab}$ depend on $(\lambda, y^a)$ and is smooth at $\lambda=0$.  The near-horizon limit is defined by rescaling $(v, \lambda, y^a)\to (v/\epsilon, \epsilon \lambda, y^a)$ and then taking the limit $\epsilon\to 0$, resulting in a metric of the same form with $\Delta, h_a, \gamma_{ab}$  replaced by their values at $\lambda=0$,  denoted by $\Delta^{(0)}, h^{(0)}_a, \gamma^{(0)}_{ab}$, which are respectively a function, 1-form and Riemannian metric on $S$. The following result does not require the detailed form of the near-horizon geometry and is typical of extremal horizons (see e.g.~\cite{Breunholder:2017ubu} in the asymptotically flat case).

\begin{lemma}\label{lem:horizon-point}
A connected component of the horizon corresponds to a point in symplectic coordinates.
\end{lemma}

\begin{proof}
For this we need to determine the change of coordinates between symplectic coordinates and GNC.    First note that the Killing fields $m_i$ must be tangent to the horizon and furthermore since they have closed orbits we may choose them to be tangent to a  cross-section $S$.  Then, one can always choose GNC that are also adapted to the toric Killing fields $m_i$. In particular, this implies $[m_i , \partial_\lambda ]=0$.

The most direct way to find the coordinate change is from the K\"ahler form which in GNC can be written as~\cite{Gutowski:2004yv}
\be
X^{(1)}= \td \lambda \wedge Z+ \lambda ( h \wedge Z- \Delta \star_3 Z)  \; ,  \label{eq:X1GNC}
\ee
where $Z= Z_a \td y^a$ is a unit norm 1-form and $\star_3$ is the Hodge star operator with respect to the metric $\gamma_{ab}$.  Closure of $X^{(1)}$ is equivalent to
\be
\hat{\td} Z= h \wedge Z- \Delta \star_3 Z+ \lambda\partial_\lambda(h \wedge Z- \Delta \star_3 Z)  \; ,
\ee
where $\hat{\td}$ is the exterior derivative tangent to constant $(v, \lambda)$ surfaces (on the horizon it is the exterior derivative on $S$).  Now, since $X^{(1)}$ is invariant under the toric symmetry, it follows that $Z$ is a $U(1)^2$-invariant 1-form; this can be seen explicitly by noting that
\be
Z= \iota_{\partial_\lambda} X^{(1)}
\ee
and since $[m_i, \partial_\lambda]=0$ the result immediately follows. Next, by definition of the symplectic coordinates (\ref{moment}) we find
\be
\td ( x_i - \lambda Z_i)= -\lambda \partial_\lambda Z_i \td \lambda +\lambda^2\partial_\lambda( h_i Z-h Z_i-\Delta \iota_{m_i}\star_3 Z)  \; ,
\ee
where we have used $\td Z= \hat{\td} Z+ \td \lambda \wedge \partial_\lambda Z$. Taking the $\lambda$ and $y^a$ components of this expression and integrating we find
\be\label{xi-NH-gen}
x_i= \lambda Z_i+ O(\lambda^2)  \; ,
\ee
where the higher order terms are smooth at $\lambda=0$ and we have set the integration constants to zero. This shows that a horizon corresponds to an isolated point in symplectic coordinates, which we take to be the origin.
\end{proof}
 
We are now ready to deduce the structure of the orbit space in symplectic coordinates. 
\begin{lemma} 
\label{thm:OS} Consider a supersymmetric solution with toric symmetry that is timelike on a neighbourhood outside the horizon (we allow the horizon to have multiple components).
If we have $N$ axis components defined by the lines $\ell_A(x) :=v^i_A x_i+c_A$, $A=1, \dots, N$, where $v_A:=v^i_A m_i=0$ define each axis component:
\begin{enumerate}
    \item The orbit space can be identified with the region (see Figure \ref{fig:orbit}) 
\be
\{ \ell_A(x) \geq 0 \, : \, A=1, \dots, N \} \subset \mathbb{R}^2 \; ,   \label{eq:orbit_sc}
\ee
where the boundary components $\ell_A(x)=0$ are axis components, and the fixed points and horizon components correspond to points of intersection of these lines. 

\item The symplectic potential takes the form
\be
g= \sum_{A=1}^N \frac{1}{2}\ell_A (x) \log \ell_A(x) + \tilde{g}
\ee
where $\tilde{g}$ is smooth on (\ref{eq:orbit_sc}) including at the lines $\ell_A(x)=0$ and the points of intersection corresponding to fixed points,  except possibly at the points of intersection corresponding to horizons.
\end{enumerate}
\end{lemma}

\begin{proof}
The general form of a supersymmetric near-horizon geometry with toric symmetry has been determined~\cite{Kunduri:2006uh}. It was found that  cross-sections of any component of the horizon must be topologically $S^3$ (or quotients) and the geometry is locally isometric to that of the CCLP black hole.  In particular, the toric symmetry degenerates at two points on the sphere (the poles) where different linear combinations of the axial Killing fields vanish.  Hence, any component of the horizon is intersected by two distinct axis components, and so by Lemma \ref{lem:horizon-point} must occur precisely at the point of intersection of the corresponding axis lines.  The remaining claims follow  from Lemma \ref{lem:axis} and  \ref{lem:fixed-point}.
\end{proof}

\begin{figure}[h!]
\centering
   \includegraphics[width=0.30\textwidth]{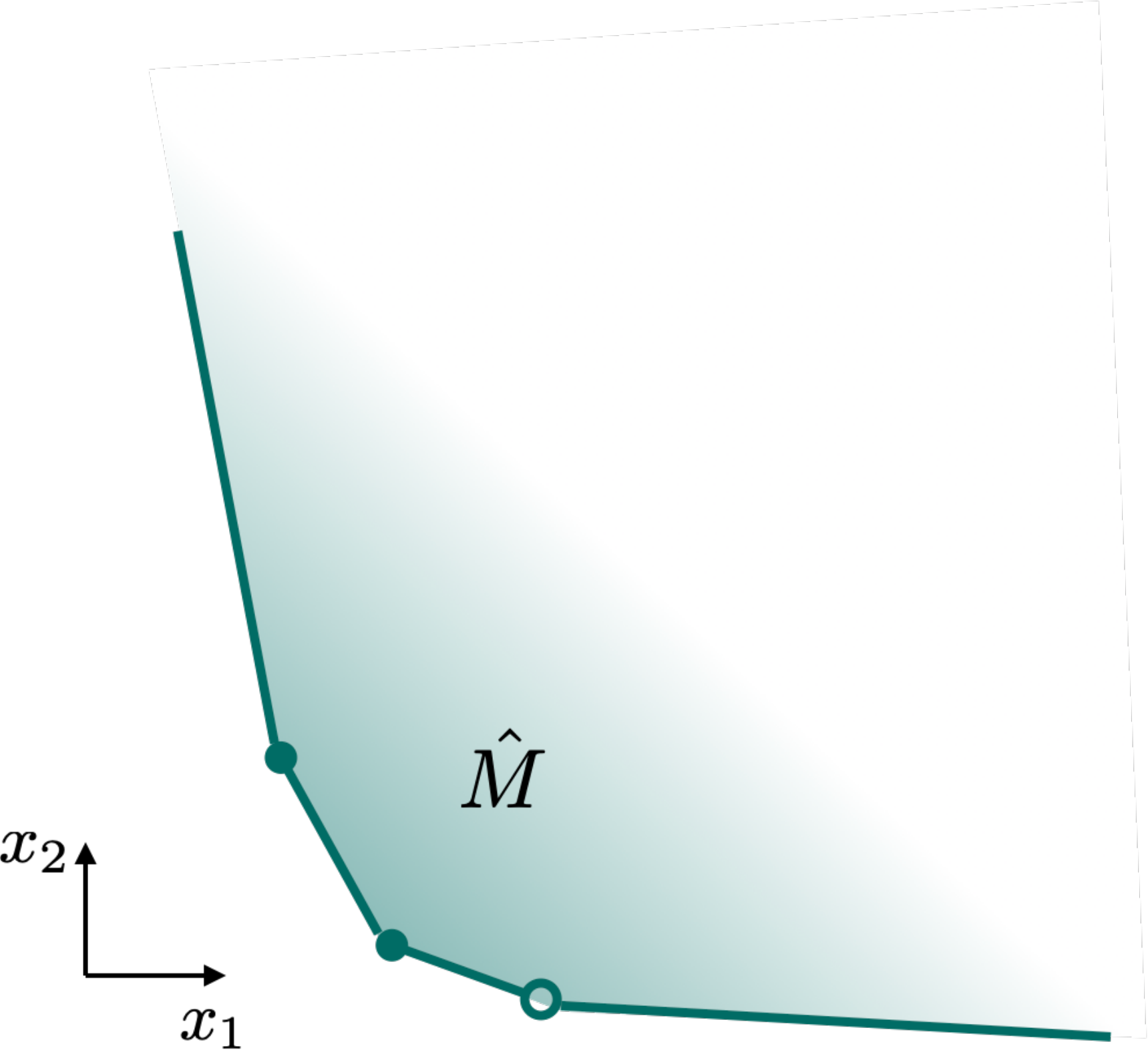}
\caption{The orbit space for a  generic toric K\"ahler base in the $x_{1}x_{2}$-plane where $x_{i}$ are moment maps associated with the axial Killing fields $m_{i}$ with $2\pi$-periodic orbits. This includes both the boundary segments and the corners which can be either horizon components (white dots) or fixed points of the torus symmetry (black dots).}
\label{fig:orbit}
\end{figure}

The above result determines the structure of the orbit space in symplectic coordinates. A schematic diagram for such orbit spaces is depicted in Figure \ref{fig:orbit}.
Furthemore, we have fully determined the behaviour of the symplectic potential near any component of the axis.  In fact we have  recovered the decomposition theorem for the symplectic potential of Abreu~\cite{Abreu} in terms of a canonical potential and a  remainder that is smooth on the axes of symmetry (note we do not assume compactness). In the next section we will use the detailed form of the general near-horizon geometry to show that the symplectic potential is similarly singular at the horizon. This will give a refinement of the second part of Lemma \ref{thm:OS}.
 
\subsection{Near-horizon geometry and general form of symplectic potential}\label{sec:NH}

In this section we will determine the general form of the symplectic potential near any connected component of a supersymmetric horizon.
The strategy is to start with the near-horizon geometry \eqref{NH-metric}, find the coordinate change between GNC and symplectic coordinates and then match to the general form for a supersymmetric toric solution (\ref{susytoric}). 

 To  this end we first note that comparing the invariants \eqref{Killinginv} yields
\be\label{omega-G-generic}
f= \lambda\Delta, \qquad \omega_i = -\frac{h_i}{\lambda \Delta^2}, \qquad G_{ij} = \lambda \Delta \left( \gamma_{ij}+ \frac{h_i h_j}{\Delta^2} \right) \; ,
\ee
where we have chosen a sign so $\Delta>0$, and inverting the matrix $G_{ij}$ we find
\be
G^{ij}= \frac{1}{\lambda \Delta} \left(\gamma^{ij}- \frac{h^i h^j}{\Delta^2+ h^k h_k} \right)  \; ,\label{Gij}
\ee
where $h^i:= \gamma^{ij} h_j$ and $\gamma^{ij}$ is the inverse matrix of $\gamma_{ij}$.

We now turn to the explicit form of the near-horizon geometry. 
The general near-horizon geometry admitting $U(1)^2$ rotational isometry and a smooth compact $S$ was determined in~\cite{Kunduri:2006uh}. We present it here in a coordinate system that also describes the special case with $SU(2)\times U(1)$ symmetry, see Appendix \ref{NH-comparison} for details.\footnote{In the notation of \cite{Kunduri:2006uh} this limit corresponds to the function $\Gamma$ being constant and therefore not being a valid coordinate.} This `unified' form of the near-horizon geometry also makes the proof of Theorem \ref{main-theorem} more transparent. The near-horizon geometry depends \emph{only} on two parameters $0< \ca^{2},\cb^{2}<1$~\footnote{Note that \eqref{kappa2} actually implies $\ca^{2},\cb^{2}<1$.} subject to 
\begin{equation}\label{kappa2}
\kappa^{2}(\ca^{2},\cb^{2})>0\,,
\end{equation}
where
\begin{equation}\label{kappa-constraint}
\kappa^{2}(\ca^{2},\cb^{2}):=-9\mathcal{A}^{4}\mathcal{B}^{4}+6\mathcal{A}^{2}\mathcal{B}^{2}(\mathcal{A}^{2}+\mathcal{B}^{2}+1)^{2}-(\mathcal{A}^{2}+\mathcal{B}^{2}+1)^{3}\Big(\mathcal{A}^{2}+\mathcal{B}^{2}-\frac{1}{3}\Big)\,.
\end{equation}
The near-horizon data explicitly reads
\begin{align}
&\Delta^{(0)}=\frac{3\kappa}{\ell\Delta_{2}(\hat{\eta})^{2}}\,, \label{NH-expl1} \\
&h^{(0)}=\frac{3\kappa\Delta_{3}(\hat{\eta})}{4\Delta_{2}(\hat{\eta})^{3}}\,\hat{\sigma}+\frac{3(\ca^{2}-\cb^{2})}{2\Delta_{2}(\hat{\eta})}\Big(\td\hat{\eta}-\frac{3\kappa}{2\Delta_{2}(\hat{\eta})^{2}}\,\hat{\tau}\Big)\,, \label{NH-expl2} \\
&\gamma^{(0)}= \frac{\ell^{2}}{12(1-\hat{\eta}^{2})\Delta_{1}(\hat{\eta})}\Big(\Delta_{2}(\hat{\eta})\,\td\hat{\eta}^{2}+\frac{3}{4}\frac{\Delta_{3}(\hat{\eta})^{2}+\kappa^{2}}{\Delta_{2}(\hat{\eta})^{2}}\,\hat{\tau}^{2}\Big) \nonumber \\
&\quad \,\, \,
+ \frac{\ell^{2}\big(4\Delta_{2}(\hat{\eta})^{3}-3\Delta_{3}(\hat{\eta})^{2}\big)}{48\Delta_{2}(\hat{\eta})^{2}}\,\hat{\sigma}^{2}+\frac{3\ell^{2}\Delta_{3}(\hat{\eta})(\ca^{2}-\cb^{2})}{8\Delta_{2}(\hat{\eta})^{2}}\,\hat{\sigma}\hat{\tau}\,, \label{NH-expl3}
\end{align}
where we have defined the 1-forms
\begin{equation}
\hat{\sigma}=\frac{1-\hat{\eta}}{\ca^{2}}\td\hat{\phi}^{1}+\frac{1+\hat{\eta}}{\cb^{2}}\td\hat{\phi}^{2}\,,\qquad\hat{\tau}=(1-\hat{\eta}^{2})\Delta_{1}(\hat{\eta})\Big(\frac{\td\hat{\phi}^{1}}{\ca^{2}}-\frac{\td\hat{\phi}^{2}}{\cb^{2}}\Big)\,,
\end{equation}
and three linear functions of $\hat{\eta}$
\begin{align}\label{Deltas}
\Delta_{1}(\hat{\eta}) & =  \frac{1+\hat{\eta}}{2}\ca^{2}+\frac{1-\hat{\eta}}{2}\cb^{2}\,,  \\
\Delta_{2}(\hat{\eta}) & =  1-\frac{1+3\hat{\eta}}{2}\ca^{2}-\frac{1-3\hat{\eta}}{2}\cb^{2}\,,\nonumber \\
\Delta_{3}(\hat{\eta}) & =  1-2\Delta_{2}(\hat{\eta})+\ca^{2}\cb^{2}-\ca^{4}-\cb^{4}\, ,   \nonumber
\end{align}
where  $(\hat{\eta}, \hat{\phi}^{i})$ are coordinates on $S$ with $-1\leq\hat{\eta}\leq1$ and ${ \hat{\phi}^{i}}\sim{ \hat{\phi}^{i}}+2\pi$ are adapted to the Killing fields $m_{i}=\partial_{ \hat{\phi}^{i}}$ and  $\Delta_1$ and $\Delta_2$ are strictly positive functions\footnote{We use hats to stress that these are coordinates valid at the horizon and they are different than the ones we will use in section \ref{sec:Calabi} to describe  the full solution (however $\partial_{ \phi^{i}}=\partial_{ \hat{\phi}^{i}}$).}. The 1-form that determines the K\"ahler form (\ref{eq:X1GNC}) is given by
\begin{equation}\label{Z0}
Z^{(0)}=\frac{\ell}{4\Delta_{2}(\hat{\eta})}\big(\kappa\,\hat{\sigma}-3(\ca^{2}-\cb^{2})\,\td\hat{\eta}\big)\,.
\end{equation}
Note that solutions with $\ca^{2}\neq\cb^{2}$ are doubly counted with the two copies related by \eqref{exchange}.  The solutions with $\ca^2=\cb^2$ have enhanced $SU(2)\times U(1)$ symmetry. 

It is now straightforward to determine the explicit near-horizon behaviour of $ \omega_i,  G_{ij}$ and $G^{ij}$. Evaluating \eqref{omega-G-generic} at leading order in $\lambda$ we find 
\begin{equation}\label{om-NH}
\omega_{i}=-\frac{\ell^{2}\Delta_{2}(\hat{\eta})}{12\kappa}\Big(\Delta_{3}(\hat{\eta})\hat{\sigma}-3(\ca^{2}-\cb^{2})\hat{\tau}\Big)_{i}\,\frac{1}{\lambda}+O(1)\,,
\end{equation}
and
\begin{equation}\label{Gdd-NH}
G_{ij}=\frac{\ell\kappa}{4\Delta_{2}(\hat{\eta})}\Big(\hat{\sigma}^{2}+\frac{\hat{\tau}^{2}}{(1-\hat{\eta}^{2})\Delta_{1}(\hat{\eta})}\Big)_{ij} \,\lambda+O(\lambda^{2})\,,
\end{equation}
with $G^{ij}$ being the inverse of \eqref{Gdd-NH}. By expressing $G^{ij}$ in the $x_{i}$ coordinates, we can extract the near-horizon behaviour of the symplectic potential. To this end, we insert \eqref{Z0} into \eqref{xi-NH-gen} to find
\begin{equation}\label{xi-NH-expl}
x_{1}=\frac{\ell\kappa}{4\Delta_{2}(\hat{\eta})}\frac{1-\hat{\eta}}{\ca^{2}}\lambda+O(\lambda^{2})\,,\qquad x_{2}=\frac{\ell\kappa}{4\Delta_{2}(\hat{\eta})}\frac{1+\hat{\eta}}{\cb^{2}}\lambda+O(\lambda^{2})\,.
\end{equation}
The inverse coordinate change to leading order is then
\begin{align}
\lambda & =\frac{2}{\ell\kappa} (\ca^{2}(1+\ca^{2}-2\cb^{2})x_{1}+\cb^{2}(1+\cb^{2}-2\ca^{2})x_{2}) +O(x^{2})\,,\nonumber \\
\hat{\eta} & =-\frac{\ca^{2}x_{1}-\cb^{2}x_{2}}{\ca^{2}x_{1}+\cb^{2}x_{2}}+O(x)\,,
\end{align}
where $O(x)$ denotes terms of order $x_i$, $O(x^2)$ terms of order $x_ix_j$ etc.  The components of $G^{ij}$ are then given by
\begin{align}\label{GUU-NH}
G^{11} & =\frac{1}{2}\Big(\frac{\ca^{4}}{\ca^{2}x_{1}+\cb^{2}x_{2}}+\frac{x_{2}}{x_{1}(x_{1}+x_{2})}\Big)+O(1)\,,\nonumber \\
G^{12} & =\frac{1}{2}\Big(\frac{\ca^{2}\cb^{2}}{\ca^{2}x_{1}+\cb^{2}x_{2}}-\frac{1}{x_{1}+x_{2}}\Big)+O(1)\,,\nonumber \\
G^{22} & =\frac{1}{2}\Big(\frac{\cb^{4}}{\ca^{2}x_{1}+\cb^{2}x_{2}}+\frac{x_{1}}{x_{2}(x_{1}+x_{2})}\Big)+O(1)\,,
\end{align}
where the $O(1)$ term is a smooth function of $(x_1, x_2)$ at the horizon $x_1=x_2=0$. These expressions can be readily  integrated to give us the symplectic potential using \eqref{hessiang}. We thus obtain
\be
g = \frac{1}{2} x_1 \log x_1+ \frac{1}{2} x_2 \log x_2 - \frac{1}{2}(x_1+x_2) \log (x_1+x_2)
+\frac{1}{2}(\ca^{2}x_{1}+\cb^{2}x_{2})\log(\ca^{2}x_{1}+\cb^{2}x_{2})
 + \tilde{g} \; , \label{NHpot}
\ee
where $\tilde{g}$ is a smooth function at the origin which arises from integrating $O(1)$ in \eqref{GUU-NH}. Recall $g$ is only defined up to an additive linear function of $x_i$ and therefore we may assume that $\tilde{g}$ vanishes quadratically in $x_i$ at the horizon.   Thus we have fixed the leading singular term in $g$ purely from the near-horizon geometry and the subleading terms are in fact $O(\lambda^2)$. 

From the above we can immediately deduce the following result which characterises the singular behaviour of the symplectic potential near any horizon.
\begin{lemma}
\label{lem:NHsymp}
Consider a supersymmetric toric solution that is timelike outside a supersymmetric horizon with compact cross-sections. Let $v_\pm= v_\pm^i m_i$ be the $2\pi$-periodic Killing fields that have fixed points on a connected component of the horizon.
 The symplectic potential takes the form
\bea
g &=& \frac{1}{2} \ell_+(x) \log \ell_+(x)+ \frac{1}{2} \ell_-(x) \log \ell_-(x) - \frac{1}{2}(\ell_+(x)+\ell_-(x)) \log (\ell_+(x)+\ell_-(x))  \nonumber
\\ &+& \frac{1}{2}(\ca^{2}\ell_{+}(x)+\cb^{2}\ell_{-}(x))\log(\ca^{2}\ell_{+}(x)+\cb^{2}\ell_{-}(x))  +\tilde{g}   \label{NHpot_2}
\eea
where $\ell_\pm (x)= v_\pm^i x_i $,  we have taken the horizon component to be at the origin, and $\tilde{g}$ is smooth at this horizon component.
\end{lemma}

We can now combine Lemma \ref{thm:OS} and Lemma \ref{lem:NHsymp} to deduce the general form for the symplectic potential.
\begin{theorem}
\label{thm:gensymp}
The symplectic potential for a supersymmetric toric solution as in Lemmas \ref{thm:OS} and \ref{lem:NHsymp} takes the form
\begin{align}
g &= \sum_{A=1}^N \frac{1}{2}\ell_A (x) \log \ell_A(x)   -  \sum_{A\in H}   \frac{1}{2}(\ell_A(x)+\ell_{A+1}(x)) \log (\ell_A(x)+\ell_{A+1}(x)) \nonumber  \\
 &+\sum_{A\in H}  \frac{1}{2}(\ca_A^{2}\ell_{A}(x)+\cb_A^{2}\ell_{A+1}(x))\log(\ca_A^{2}\ell_{A}(x)+\cb_A^{2}\ell_{A+1}(x))  +\tilde{g}  \; ,
\end{align}
where $A\in H$ iff  the point defined by $\ell_A(x)=\ell_{A+1}(x)=0$ corresponds to a component of the horizon and $\ca_A, \cb_A$ are parameters of the corresponding near-horizon geometry, and $\tilde{g}$ is smooth on (\ref{eq:orbit_sc}) including at the boundaries and corners.
\end{theorem}

Recall that in Lemma \ref{thm:OS}  we assumed that a single symplectic chart contained all components of the axis.  As discussed at the start of section \ref{ax-hor-orb}, we anticipate that under certain reasonable global assumptions (including the asymptotics) one should be able to prove that symplectic coordinates provide a global chart of the exterior region of a black hole spacetime.  
Theorem \ref{thm:gensymp} then gives the general form for a symplectic potential for a supersymmetric AdS black hole spacetime with toric symmetry of arbitrary topology.  It is an analogue of Abreu's result for compact toric K\"ahler manifolds~\cite{Abreu}.

\subsection{Examples}
\label{sec:ex}

In this section we will first list the known solutions and their symplectic potentials. Then using the above results we will write down the form of the symplectic potential for possible new black hole solutions with nontrivial topology.

\subsubsection{AdS$_5$ and known black holes}

The simplest relevant examples of  toric K\"ahler metrics are  given by the $SU(2)\times U(1)$ invariant K\"ahler metrics~\footnote{Compared to~\cite{Lucietti:2021bbh} we have rescaled $r_{\text{there}}=2\alpha\rcclp_{\text{here}}$.}
\begin{align}\label{GRmetr}
h &= \frac{\td \rcclp^2}{V(\rcclp)} + \alpha^{2}\rcclp^{2} ( \td \vartheta^2+ \sin^2\vartheta \td \varphi^2)+ 4\alpha^{4}\rcclp^{2} V(\rcclp) (\td \psi+ \cos \vartheta \td\varphi)^2 \,,\nonumber\\
X^{(1)}&= \td  \left(\alpha^{2} \rcclp^2 (\td \psi+ \cos \vartheta \td\varphi) \right)  \; ,
\end{align}
where $V(\rcclp)>0$ is an arbitrary function, $\alpha$ is a constant, and $(\vartheta, \psi, \varphi)$ are Euler coordinates on $S^3$.
This family includes the  important case of the Bergmann metric for  $V=1+\frac{\rcclp^2}{\ell^2}$ and $\alpha=1/2$ which is the base space of global AdS$_5$ (normalised so $f=1$). For $\alpha>1/2$ this class  includes the base of the GR black hole which also has $V= 1+ \frac{\rcclp^2}{\ell^2}$, or its near-horizon geometry which has $V=1$ (also a supersymmetric solution).

We now write these solutions in symplectic coordinates.
It is convenient to define $2\pi$-periodic angles $\phi^i$ by $\psi=\phi^1+\phi^2$ and $\varphi= -\phi^1+\phi^2$ so that $\partial_{\phi^1} = \partial_\psi-\partial_\varphi$ and $\partial_{\phi^2} = \partial_\psi+\partial_\varphi$.  Then, from (\ref{moment}) we deduce that the symplectic coordinates adapted to $m_i=\partial_{\phi^i}$ are given by
\be
x_1 =2\alpha^{2} \rcclp^2 \sin^2(\vartheta/2), \qquad x_2 = 2\alpha^{2} \rcclp^2  \cos^2(\vartheta /2) \; .
\ee
Thus the component of the axis $\vartheta=0$ ($r>0$) corresponds to $x_1=0$ on which $\partial_{\phi^1}=0$, and the axis $\vartheta=\pi$ ($r>0$) corresponds to $x_2=0$ on which $\partial_{\phi^2}=0$. The interior of the orbit space is $x_1>0, x_2>0$ and the asymptotic region is  $x_1+x_2 \to \infty$. Thus the image of the moment maps is the upper right quarter $x_1x_2$-plane. The corner $x_1=x_2=0$ ($r=0$) is where the two axis components meet and the torus symmetry has a fixed point.  Now, comparing to (\ref{h_toric}) and (\ref{hessiang}), we find that the symplectic potential for the Bergmann metric is
\be
 g_{\text{B}} =\frac{1}{2} x_1 \log x_1+ \frac{1}{2} x_2 \log x_2  - \frac{1}{2} (x_1+x_2+ \tfrac{1}{2} \ell^2) \log (x_1+x_2+ \tfrac{1}{2} \ell^2)  \; . \label{B_pot}
 \ee
 Notice this is related to the symplectic potential of the Fubini-Study metric on $\mathbb{CP}^2$ by an analytic continuation $\ell^2\to -\ell^2$.
 
 For the GR black hole we find the symplectic potentials takes a remarkably similar form
  \bea\label{GR_pot}
  g_{\text{GR}} &=& \frac{1}{2} x_1 \log x_1+ \frac{1}{2} x_2 \log x_2+ \left( \frac{1}{8\alpha^2} - \frac{1}{2} \right) (x_1+x_2) \log (x_1+x_2)   \nonumber \\ &&- \frac{1}{8\alpha^2} (x_1+x_2+ 2 \alpha^2 \ell^2) \log (x_1+x_2+ 2 \alpha^2 \ell^2)   \; .
 \eea
 Observe that for $\alpha=1/2$ this reduces to that for the Bergmann metric (as it should).  For the near-horizon geometry of the GR black hole we find
  \be\label{NHGR_pot}
  g_{\text{NHGR}} =   \frac{1}{2} x_1 \log x_1+ \frac{1}{2} x_2 \log x_2 + \left( \frac{1}{8\alpha^2} - \frac{1}{2} \right) (x_1+x_2) \log (x_1+x_2)\; ,
 \ee
 which is the same as the black hole potential (\ref{GR_pot})  with the last term omitted.
 
 We now consider more generic solutions with toric symmetry. The CCLP black hole (as well as its near-horizon geometry) is the most general known black hole solution in this symmetry class. We write it in a compact form in terms of parameters $A$ and $B$ in  Appendix \ref{sec:CCLP}.  The metric and the K\"ahler form for the base are given by \eqref{CCLPmetr} and the relation with the $2\pi$-periodic angles in \eqref{CCLPangles}. The GR solution is a special case of CCLP for $A=B=(2\alpha)^{-1}$, but note that we must rescale  $(\psi,\varphi)_{\text{GR}}=(A^{2}\psi,A^{2}\varphi)_{\text{CCLP}}$. We find that the symplectic coordinates are given by 
\be
x_1 =\frac{1}{2 A^{2}} \rcclp^2 \sin^2(\vartheta/2), \qquad x_2 = \frac{1}{2 B^{2}} \rcclp^2  \cos^2(\vartheta/2)  \; ,
\ee
with the axis- (and fixed points-) structure being identical to the $SU(2)\times U(1)$ case, while the asymptotic region corresponds to $A^{2}x_1+B^{2}x_2 \to \infty$. The symplectic potential of CCLP can again be evaluated from (\ref{h_toric}) and (\ref{hessiang}) and takes the remarkably simple and explicit form
\begin{align}\label{CCLP_pot}
g_{\text{CCLP}} & =\frac{1}{2}x_{1}\log x_{1}+\frac{1}{2}x_{2}\log x_{2}-\frac{1}{2}(x_{1}+x_{2})\log(x_{1}+x_{2})+\frac{1}{2}(A^{2}x_{1}+B^{2}x_{2})\log(A^{2}x_{1}+B^{2}x_{2})\nonumber \\
 & -\frac{1}{2}(A^{2}x_{1}+B^{2}x_{2}+\tfrac{1}{2} \ell^2)\log(A^{2}x_{1}+B^{2}x_{2}+\tfrac{1}{2} \ell^2)\, .
\end{align}
The potential for the near-horizon geometry of the CCLP black hole is
\be\label{NHCCLP_pot}
g_{\text{NHCCLP}} =\frac{1}{2}x_{1}\log x_{1}+\frac{1}{2}x_{2}\log x_{2}-\frac{1}{2}(x_{1}+x_{2})\log(x_{1}+x_{2})+\frac{1}{2}(A^{2}x_{1}+B^{2}x_{2})\log(A^{2}x_{1}+B^{2}x_{2})\, ,
\ee
which is  also the same as that of its parent black hole \eqref{CCLP_pot} with the last term omitted. Observe that these expressions reduce to those for the GR black hole for $A=B=(2\alpha)^{-1}$ as they must.

It is worth emphasising that the singular part of the symplectic potential for the above solutions is as predicted by our general Theorem \ref{thm:gensymp}.
It is interesting to note that for all these black hole solutions the symplectic potential takes a form similar to that for the canonical metric on compact toric K\"ahler manifolds~\cite{Abreu}.

\subsubsection{Symplectic potentials for black holes with nontrivial topology}

We are interested in spacetimes asymptotic to global AdS$_5$.  For such spacetimes we will choose the toric Killing fields $m_i$ to be  $2\pi$-periodic and orthogonal on the $S^3$ at infinity.  In particular, in terms of symplectic coordinates adapted to such $m_i$ the symplectic potential for global AdS$_5$ is given by (\ref{B_pot}). Therefore,  the axis will always include two semi-infinite axes, say $x_1=0$ and $x_2=0$, for large enough $x_2$ and $x_1$ respectively.  It is helpful to note that 
\be
g_{\mathbb{R}^4}=  \frac{1}{2} x_1 \log x_1+ \frac{1}{2} x_2 \log x_2
\ee
is the symplectic potential for euclidean space $\mathbb{R}^4$ (note that this is inline with Lemma \ref{lem:fixed-point}).  This has the same axis structure (and topology) as the Bergmann metric  and its symplectic potential (\ref{B_pot}) differs only in the third term which is smooth at the axes and  dictates the asymptotics.    We will now deduce the general form of the symplectic potential assuming there are at least two semi-infinite axis components.

As a warm up let us first consider the simplest case where we have two axis components $x_1=0$ and $x_2=0$ and a horizon at the origin. This includes the known CCLP black hole. From Theorem \ref{thm:gensymp} we deduce that we can write the potential as
\be
g =  \frac{1}{2} x_1 \log x_1+ \frac{1}{2} x_2 \log x_2 - \frac{1}{2}(x_1+x_2) \log (x_1+x_2) 
+ \frac{1}{2}( \ca^2 x_1+ \cb^2 x_2) \log ( \ca^2 x_1+ \cb^2 x_2)+ \tilde{g}
\ee
where $\tilde{g}$ is smooth at the axes $\{ x_1=0, x_2>0 \}$ and $\{ x_1>0, x_2=0 \}$, and at the horizon $x_1=x_2=0$.  Inspecting the symplectic potential for the CCLP black hole \eqref{CCLP_pot} we see that
\be
\tilde{g}_{\text{\tiny CCLP}}= - \frac{1}{2}( \ca^2 x_1+ \cb^2 x_2+ \tfrac{1}{2} \ell^2)\log (\ca^2 x_1+ \cb^2 x_2+ \tfrac{1}{2} \ell^2) \; ,
\ee
which is indeed smooth at the axes and horizon inline with Theorem \ref{thm:gensymp}.  On the other hand $\tilde{g}=0$  gives a solution corresponding to the near-horizon geometry of CCLP.  

 Of course,  Theorem \ref{thm:gensymp}  allows us to write down the potential for any horizon and axis structure.  
For simplicity we will only consider the next simplest example where we have three axis components: the two semi-infinite axes $\ell_1=0$, $\ell_3=0$, and a finite axis $\ell_2=0$ which joins the two: 
\be
\ell_1:=x_1, \qquad \ell_2:= p x_1+ q x_2- a, \qquad \ell_3:=x_2  \; ,   \label{eq:ell2}
\ee  
where $(p,q)$ are coprime integers.
In the $(m_1, m_2)$ basis these correspond to the vanishing of the vectors $v_1=(1,0)$, $v_2=(p, q)$ and $v_3=(0,1)$ respectively.  The intersection of the finite axis with the semi-infinite ones correspond to corners of the orbit space and these must be on the positive $x_1$ and $x_2$ axes which requires $a/p>0$ and $a/q>0$. Thus we may assume $p>0, q>0, a>0$.   There are three cases to consider depending on if one, two or none of the corners are horizons. 

\begin{enumerate}[label=(\alph*)]
\item{\bf  Black lens}. First suppose the corner $x_1=0, x_2=a/q$ corresponds to a horizon and the other corner to a fixed point as in Figure \ref{blacklens}. Then absence of an orbifold singularity (\ref{orbifold}) at the fixed point requires $\det (v_2, v_3)= -p=\pm 1$ so without loss of generality we may always set $p=1$.  The horizon topology is determined by $\det(v_1, v_2)=q$ and is a lens space $L(q,1)$. Then, from Theorem \ref{thm:gensymp}, we deduce that the general form of the symplectic potential is 
\begin{align}
g &=  \frac{1}{2} x_1 \log x_1+ \frac{1}{2} \ell_2 \log \ell_2 +  \frac{1}{2} x_2 \log x_2 \nonumber  \\ &- \frac{1}{2}(x_1+\ell_2) \log (x_1+\ell_2)   + \frac{1}{2}( \ca^2 x_1+ \cb^2 \ell_2) \log ( \ca^2 x_1+ \cb^2 \ell_2)
+ \tilde{g}  \; ,
\end{align}
where $\ell_2=x_1+q x_2-a$ and   $\tilde{g}$ is smooth at each of the axes $x_1=0, x_2=0$ and $\ell_2=0$ and at both corners.   By construction this gives the symplectic potential for a regular black lens spacetime.  An asymptotically flat supersymmetric $L(2,1)$ black lens with this axis and horizon structure is known~\cite{Kunduri:2014kja}.

\item {\bf Double-black hole}. Next suppose both corners are horizons so we have a double-black hole as in Figure \ref{doubleblackhole}.  This gives an $L(q,1)$ black lens at $x_1=0, x_2=a/q$ and an $L(p,1)$ black lens at $x_2=0, x_1=a/p$ with no restriction on $p,q$ (other than they are coprime).  The symplectic potential now takes the form
\bea
g &=&  \frac{1}{2} x_1 \log x_1+ \frac{1}{2} \ell_2 \log \ell_2 +  \frac{1}{2} x_2 \log x_2 \nonumber \\ \nonumber  &-& \frac{1}{2}(x_1+\ell_2) \log (x_1+\ell_2)  - \frac{1}{2}(x_2+\ell_2) \log (x_2+\ell_2)  \\ 
&+& \frac{1}{2}( \ca^2 x_1+ \cb^2 \ell_2) \log ( \ca^2 x_1+ \cb^2 \ell_2)+  \frac{1}{2}(\tilde{\ca}^2 \ell_2+ \tilde{\cb}^2 x_2) \log ( \tilde{\ca}^2 
\ell_2+ \tilde{\cb}^2 x_2)
+ \tilde{g}
\eea
where $\ell_2$ is given by (\ref{eq:ell2}) and there are singular terms from the horizons at   $x_1=0, x_2=a/q$ and  $x_2=0, x_1=a/p$ respectively (we distinguish the near-horizon parameters of the second horizon by tildes). Again $\tilde{g}$ is smooth at all the axis lines and corners.   Note the special case $p=q=1$ corresponds to a double $S^3$ black hole.

\item {\bf Soliton}. 
Now suppose both corners are fixed points of the toric symmetry as in Figure \ref{soliton}. Absence of orbifold singularities at the corners now requires $p=q=1$ (fixing signs).   The symplectic potential in this case is simply
\be
g= \frac{1}{2} x_1 \log x_1+ \frac{1}{2}(x_1+x_2-a) \log (x_1+x_2-a) +  \frac{1}{2} x_2 \log x_2  + \tilde{g}
\ee
where $\tilde{g}$ is smooth at the three axes.  This gives a soliton with a bolt at the finite axis.
\end{enumerate}

\begin{figure}[h]
\centering
\begin{subfigure}{0.3\textwidth}
\includegraphics[width=4cm, height=5cm]{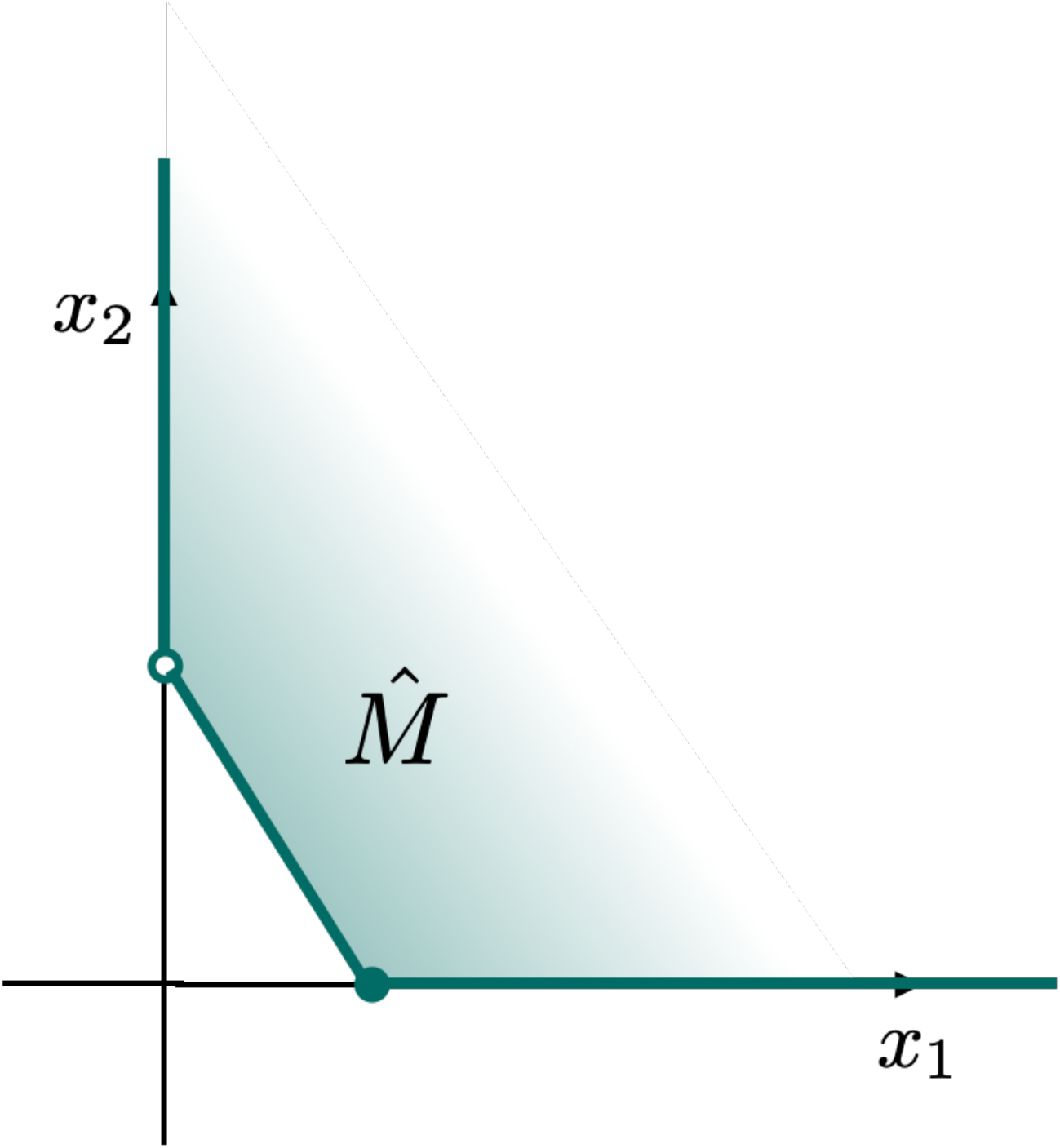} 
\caption{Black lens}
\label{blacklens}
\end{subfigure}
\begin{subfigure}{0.3\textwidth}
\includegraphics[width=4cm, height=5cm]{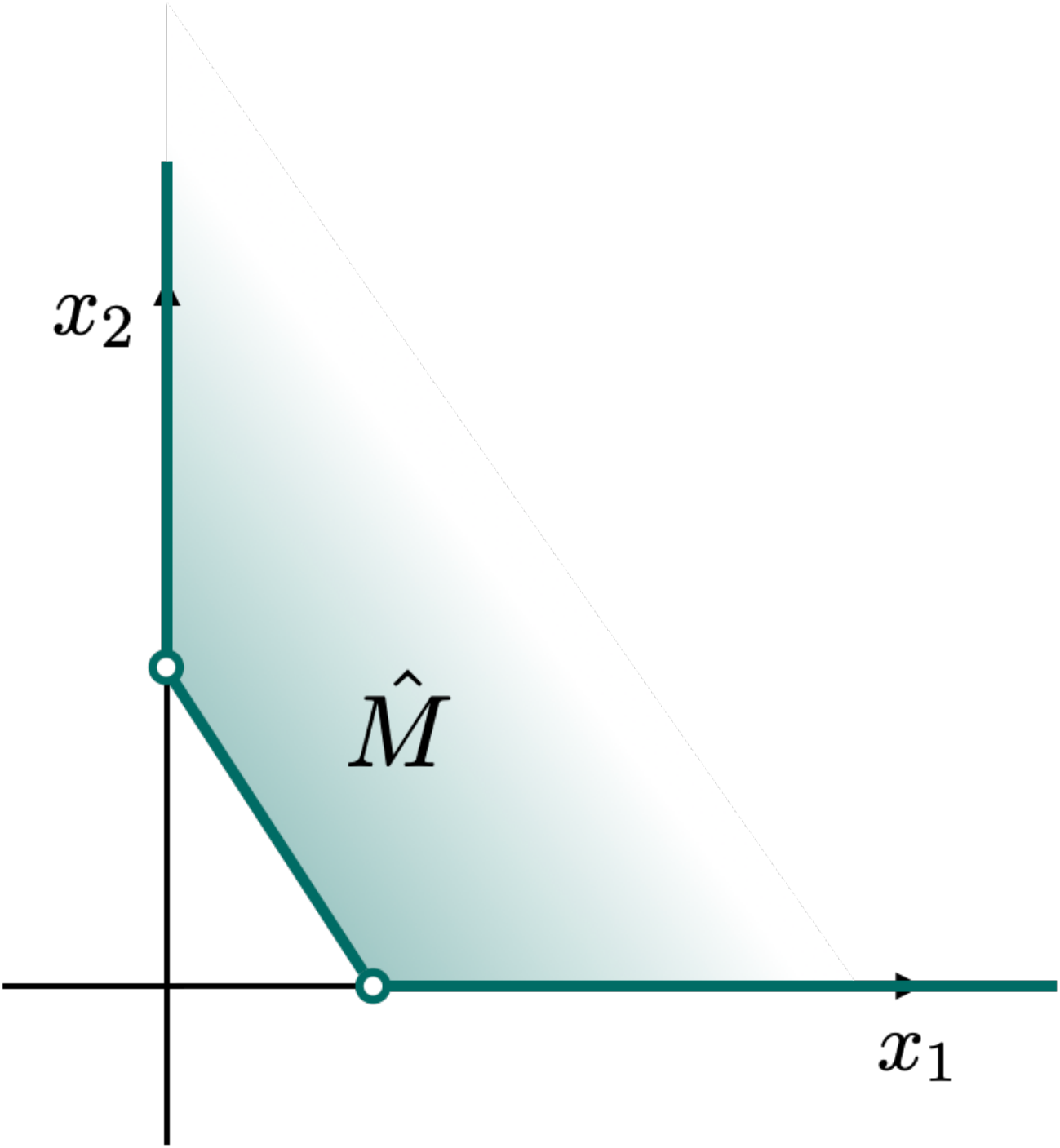}
\caption{Double black hole}
\label{doubleblackhole}
\end{subfigure}
\begin{subfigure}{0.3\textwidth}
\includegraphics[width=4cm, height=5cm]{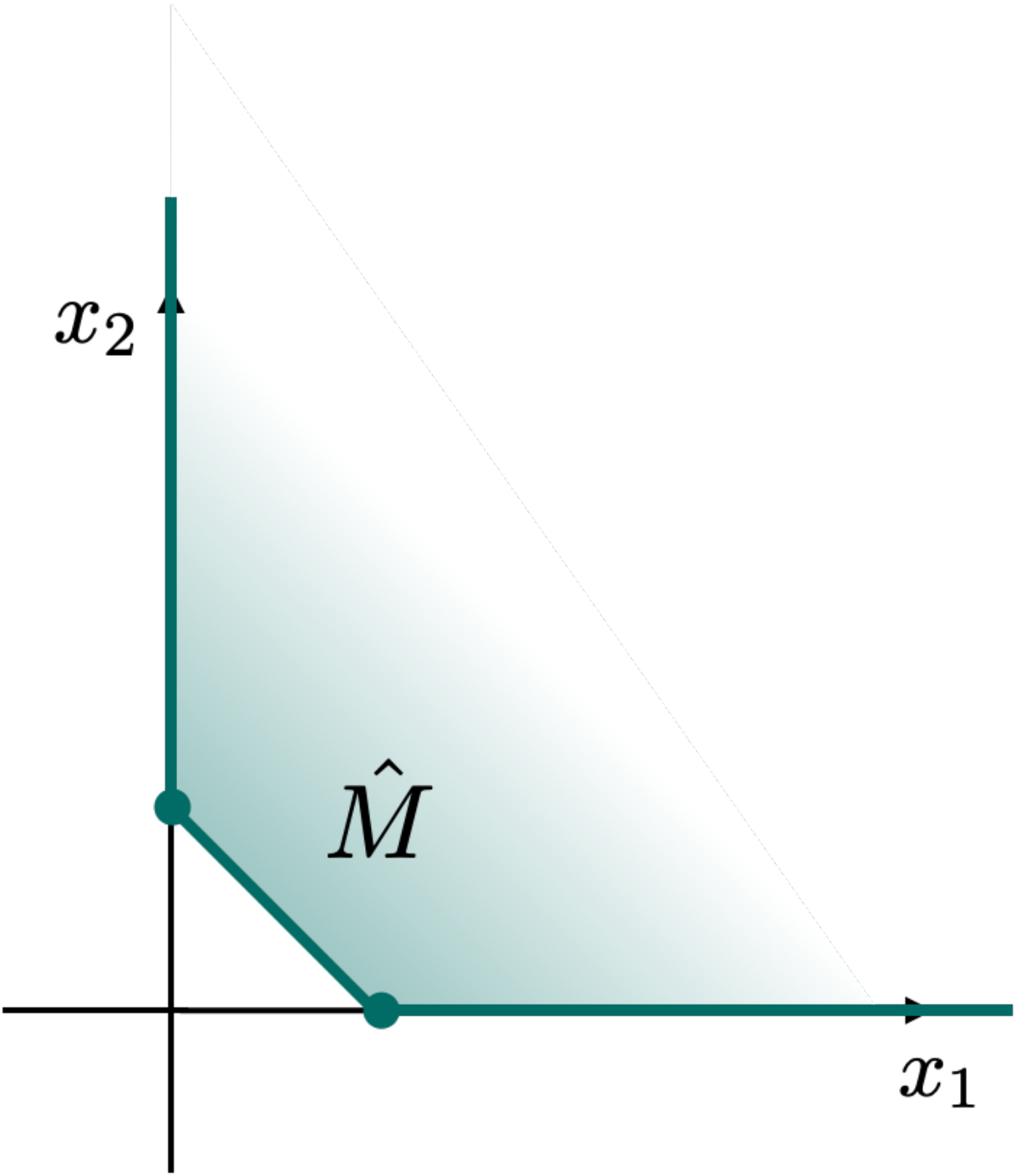}
\caption{Soliton}
\label{soliton}
\end{subfigure}
\caption{The three cases of orbit space with three axis components depending on whether we have 1, 2 or 0 horizon components (white dots).}
\label{fig:image2}
\end{figure}

We emphasise that there are no known solutions to (\ref{PDE}) that realise the above black lens and double-black hole axis and horizon structures (i.e. case (a) and (b) above). An asymptotically AdS$_5/\mathbb{Z}_p$ soliton spacetime with $SU(2)\times U(1)$ symmetry is known with an axis structure as in case (c) and therefore gives example for this topology type~\cite{Lucietti:2021bbh} (its K\"ahler base is of the form (\ref{GRmetr}), although it has a more complicated $V(r)$). However, the only known asymptotically AdS$_5$ soliton spacetime~\cite{Chong:2005hr, Cassani:2015upa, Durgut:2021rma}, possesses an ergosurface where $V$ is null on a timelike hypersurface so the K\"ahler base metric is not globally defined in this case (the region outside to the ergosurface possesses an extra boundary corresponding to the ergosurface and the symplectic potential is much more complicated).

The existence of a black lens or a multi-black hole in AdS$_5$ remains a central open problem in the classification of AdS$_5$ black hole spacetimes.  The above shows that this problem reduces to finding the smooth part of the symplectic potential that solves (\ref{PDE}) and gives an asymptotically AdS$_5$ spacetime. In the absence of extra structure or symmetries this appears to be a formidable problem given the complexity of (\ref{PDE}).  In the next section we identify a special subclass of toric solutions that contains the CCLP black hole for which we are able to completely solve the classification problem.

\section{Supersymmetric toric solutions of Calabi type}\label{sec:Calabi}

Section \ref{sec:toric} was devoted to an analysis of generic supersymmetric toric solutions in the timelike class. In particular, we showed that such solutions can be described in symplectic coordinates by a symplectic potential that must obey a complicated 8th order non-linear PDE (its explicit form is unenlightening and can be obtained by inserting \eqref{h_toric} and \eqref{hessiang} in \eqref{PDE}). Furthermore, we derived the behaviour of the symplectic potential near any component of a horizon or axis.  Unfortunately, it appears that a classification of K\"ahler geometries that satisfy this PDE (even under our boundary conditions) is at present out of reach.

In this section we will focus on a particular subclass of toric K\"ahler bases that are of Calabi type~\cite{Apostolov2001TheGO} (we will define this in subsection \ref{sec:Calabi-general}). Our motivation for considering this class of solutions is that the CCLP black hole, which is the most general known black hole in this theory, has a K\"ahler base of Calabi type~\cite{Cassani:2015upa}.  Furthermore, we will show that the most general near-horizon geometry (with compact cross-sections) discussed in section \ref{sec:NH} is also of  Calabi type (see subsection \ref{sec:Calabi-NH}), which implies that any supersymmetric black hole in this theory must be of Calabi type at least near the horizon.   We will find that for this class of K\"ahler bases we can derive the general black hole solution and that this agrees precisely with the CCLP black hole. Therefore, we obtain a constructive uniqueness theorem for this black hole solution.

\subsection{Toric K\"ahler metrics of Calabi type}\label{sec:Calabi-general}

In this subsection, we will introduce toric K\"ahler manifolds of Calabi type. A general definition of K\"ahler surfaces of Calabi type was given by Apostolov, Calderbank and Gauduchon in their study of K\"ahler surfaces admitting a hamiltonian 2-form, the latter a concept which they also introduce~\cite{Apostolov2001TheGO}\footnote{A 2-form is hamiltonian if it is closed, invariant under the complex structure, and its traceless part is a twistor form.}. They find that K\"ahler surfaces admitting a hamiltonian 2-form necessarily possess two commuting hamiltonian Killing vectors, and if these  vectors are independent the K\"ahler surface is orthotoric, whereas if they are dependent it is of Calabi type.  An orthotoric K\"ahler surface is one that admits two linearly independent hamiltonian Killing fields with moment maps $\xi+\eta$ and $\xi \eta$ such that $\td \xi$ and $\td \eta$ are orthogonal.  We will give our own analogous definition of Calabi type in the context of toric K\"ahler surfaces (which is a subclass of Calabi type according to~\cite{Apostolov2001TheGO}).

\begin{definition}
\label{def}
A  K\"ahler surface is toric and of Calabi type if it admits two linearly independent hamiltonian Killing fields with moment maps $\rho$ and $\rho \eta$ such that $\td \rho$ and $\td \eta$ are orthogonal.
\end{definition}

Let us work out some basic consequences of the above definition. It is convenient to work in a $2\pi$-periodic basis of toric Killing fields
 $m_{i}=\partial_{\phi^{i}}$, $\phi^{i}\sim\phi^{i}+2\pi$. Then according to this definition there exist linear combinations of the Killing fields $k_i= A_{i}^{~j}m_j$ where $A\in GL(2, \mathbb{R})$ with moment maps $\rho, \rho \eta$ such that 
 \begin{equation}\label{rho-eta-orthoginal}
 \td \rho\cdot \td\eta=0\,.
\end{equation}
Note that the Killing fields $k_1, k_2$ need not have closed orbits. It is useful to adapt coordinates so $k_1=\partial_\psi$ and $k_2=\partial_\varphi$ correspond to the moment maps $\rho$ and $\rho \eta$ respectively, and parameterise the $GL(2, \mathbb{R})$ transformation explicitly in terms of these coordinates as
\begin{equation}\label{eq:angles-Calabi}
\psi=a_{i}\phi^{i}\,,\qquad\varphi=b_{i}\phi^{i}\,,\qquad\text{with}\qquad  \langle a,b\rangle\neq 0\, , 
\end{equation}
where $a_i, b_i$ are constants and we have introduced the notation $ \langle a,b\rangle:=\epsilon^{ij}a_{i}b_{j}$ with $\epsilon^{ij}$ antisymmetric and $\epsilon^{12}=1$.
 Since the K\"ahler form \eqref{KahlerForm-gen} is invariant under such $GL(2, \mathbb{R})$ transformations, we have
\begin{equation}\label{X1-Calabi}
X^{(1)} = \td \big( \rho (\td \psi+ \eta \td \varphi)\big),
\end{equation}
while the moment maps (\ref{moment}) for the Killing vectors with $2\pi$-periodic orbits $m_{i}$ are 
\begin{equation}
x_{i}=c_{i}+\rho(a_{i}+b_{i}\eta)\,,\label{eq:Calabi-moment-maps}
\end{equation}
where $c_{i}$ are constants of integration which we will leave arbitrary for later convenience.

Now we will show that toric K\"ahler manifolds of Calabi type can be described in terms of two functions $F(\rho), G(\eta)$ as the following proposition shows.
\begin{prop}
A  toric K\"ahler surface is of Calabi type iff its metric can be written in the form 
\begin{equation} \label{h-Calabi}
h = \frac{\rho}{F(\rho)} \td \rho^2+ \frac{F(\rho)}{\rho} ( \td \psi+ \eta \td \varphi )^2+ \rho \left( \frac{\td \eta^2}{G(\eta)}+ G(\eta) \td \varphi^2 \right)\,,
\end{equation}
where $F(\rho), G(\eta)$  are arbitrary functions and $\psi,\varphi$ are given by \eqref{eq:angles-Calabi}.
\end{prop}

\begin{proof} The proof makes use of the symplectic potential and of the generic form \eqref{h_toric} for a toric K\"ahler metric.  First let us assume that we have a toric K\"ahler metric of Calabi type. Using \eqref{eq:Calabi-moment-maps} the orthogonality condition \eqref{rho-eta-orthoginal} becomes the following condition for the symplectic potential
\begin{equation}
0=h_{\rho\eta}=\rho(a_{i}+b_{i}\eta)b_{j}G^{ij}=\rho(a_{i}+b_{i}\eta)b_{j}\partial^{i}\partial^{j}g=\rho\partial_{\rho}(\rho^{-1}\partial_{\eta}g)\,,\label{eq:GrhoetaAux}
\end{equation}
which can be easily integrated to give 
\begin{equation}\label{eq:Calabi-sympl-potential}
g=g_{0}(\rho)+\rho g_{1}(\eta)\,,
\end{equation}
where $g_0(\rho), g_1(\eta)$ are arbitrary functions. A  computation then reveals that $h_{\rho\rho}=g_{0}''(\rho)$ , $h_{\eta\eta}=\rho g_{1}''(\eta)$ and hence defining 
\begin{equation}
F(\rho) :=\frac{\rho}{g_{0}''(\rho)}\,,\qquad G(\eta):=\frac{1}{g_{1}''(\eta)}\,,\label{eq:FGofAB}
\end{equation}
the $(\rho, \eta)$ components of the metric are as claimed in  \eqref{h-Calabi}. To find the $(\psi, \varphi)$ components we can use that the $(\phi^i ,\phi^j)$ components are given by $G_{ij}$ and that this is the inverse matrix of the $(x_i, x_j)$ components  $G^{ij}=\partial^i \partial^j g$ which we now know.  A short computation then confirms that the metric indeed takes the form (\ref{h-Calabi}). Finally, it is easy to show the converse statement, namely that \eqref{h-Calabi} implies the Calabi property\footnote{In fact \eqref{h-Calabi} is enough to show the K\"ahler property. The K\"ahler form is given by \eqref{X1-Calabi}.}. To this end, notice that  \eqref{h-Calabi}, \eqref{X1-Calabi} have the form \eqref{h_toric},\eqref{KahlerForm-gen}  for moment maps $\rho, \rho\eta$ and the condition \eqref{rho-eta-orthoginal} is obviously satisfied.
\end{proof}

It is useful to  note that from the above proof we can extract the Gram matrix of Killing fields  in the original $2\pi$-periodic coordinates $\phi^{i}$, which we find is,
\begin{equation}\label{eq:GDD-Calabi}
G_{ij}=\frac{F(\rho)}{\rho}(a_{i}+b_{i}\eta)(a_{j}+b_{j}\eta)+\rho G(\eta)b_{i}b_{j}\,  .
\end{equation}
From this we can also deduce the following useful projections
\begin{equation}\label{Gproj}
\frac{G_{ij}\epsilon^{ik}\epsilon^{jl}b_{k}b_{l}}{\langle a,b\rangle^{2}}=\frac{F(\rho)}{\rho}\,,\qquad\frac{G_{ij}\epsilon^{ik}\epsilon^{jl}a_{k}b_{l}}{\langle a,b\rangle^{2}}=-\frac{\eta F(\rho)}{\rho}\,,\qquad\frac{G_{ij}\epsilon^{ik}\epsilon^{jl}a_{k}a_{l}}{\langle a,b\rangle^{2}}=\frac{\eta^2 F(\rho)}{\rho}+\rho G(\eta)\, ,
\end{equation}
which  in particular allow us to express $F(\rho), G(\eta)$ in terms of invariants of the K\"ahler base.   In fact, it is worth noting the following useful characterisation of toric K\"ahler metrics of Calabi type.
\begin{lemma}
\label{lem:calabi}
A toric K\"ahler metric with symplectic coordinates \eqref{eq:Calabi-moment-maps} and a Gram matrix of Killing fields $G_{ij}$ given by \eqref{eq:GDD-Calabi} for some $a_i, b_i, c_i, \rho, \eta, F(\rho), G(\eta)$,  must be of Calabi type \eqref{h-Calabi}.
\end{lemma}

\begin{proof}
Change basis for the Killing fields by introducing coordinates  $\tilde{\phi}^1:= \psi, \tilde{\phi}^2:= \varphi$ defined by \eqref{eq:angles-Calabi}. Then the Gram matrix  \eqref{eq:GDD-Calabi}  in this new basis $\tilde{G}_{ij}$  coincides with the  $(\psi, \varphi)$  components of the metric \eqref{h-Calabi}.  The inverse matrix $\tilde{G}^{ij}$ is thus determined in terms of $\rho, \eta, F(\rho), G(\eta)$.  But $G^{ij}\td x_i \td x_j = \tilde{G}^{ij} \td \tilde{x}_i \td \tilde{x}_j$ where the symplectic coordinates $\tilde{x}_1= \rho$, $\tilde{x}_2 = \rho \eta$ and a short computation using the already determined $\tilde{G}^{ij}$ gives the $(\rho, \eta)$ components of  \eqref{h-Calabi} as required.
\end{proof}

For later use it is convenient to record the expressions for the Ricci potential and the Ricci scalar of the K\"ahler metrics (\ref{h-Calabi}) which are given by
\begin{equation}\label{Cal-useful}
P=-\frac{1}{2}\frac{F'(\rho)}{\rho}(\td\psi+\eta\td\varphi)-\frac{1}{2}G'(\eta)\td\varphi\,\qquad \qquad
R=-\frac{F''(\rho)+G''(\eta)}{\rho}\, ,
\end{equation}
respectively.

The parametrisation of the Calabi metric (\ref{h-Calabi}) (and hence the functions $F(\rho), G(\eta)$)  and the  K\"ahler form \eqref{X1-Calabi}  in the coordinates $(\rho,\eta,\psi,\varphi)$ is far from unique. In fact, the form of \eqref{h-Calabi} and \eqref{X1-Calabi} remains invariant under the following constant rescalings of $\rho$ and affine transformations of $\eta$:
\begin{equation}\label{Calabi-symmetries}
\rho\to K_{\rho}\rho\,,\qquad\psi\to K_{\rho}^{-1}(\psi-C_{\eta}K_{\eta}^{-1}\varphi)\,,\qquad\eta\to K_{\eta}\eta+C_{\eta}\,,\qquad\varphi\to K_{\rho}^{-1}K_{\eta}^{-1}\varphi\,,
\end{equation}
where $K_{\rho},K_{\eta}$ are non-zero real constants and $C_{\eta}\in\mathbb{R}$ \footnote{Of course, the functions should also transform according to  $F(\rho)\to K_{\rho}^{3}F(\rho)$ and  $G(\eta)\to K_{\rho}K_{\eta}^{2}G(\eta)$.
}.
In turn, \eqref{eq:angles-Calabi} imply that these transformations also act on $a_{i}$ and $b_{i}$ as 
\begin{equation}\label{ab-transform}
a_{i}\to K_{\rho}^{-1}(a_{i}-C_{\eta}K_{\eta}^{-1}b_{i})\,,\qquad b_{i}\to K_{\rho}^{-1}K_{\eta}^{-1}b_{i}\,.
\end{equation}
Note that these transformations include discrete symmetries which flip independently the signs of $\rho$ and $\eta$.

From the five-dimensional point of view, we are really only interested in the K\"ahler base metric up to the rescalings \eqref{time-resc}. For (\ref{h-Calabi}) these rescalings can be realised by
\begin{equation}\label{time-resc2}
\rho\to K^{-1}\rho\,,\qquad F(\rho)\to K^{-2}F(\rho)\,,
\end{equation}
with $\eta,\psi,\varphi$ and $G(\eta)$ unchanged. It is easy to see that the metric and the K\"ahler form are rescaled while the Calabi property is preserved. 

\subsection{Near-horizon analysis}\label{sec:Calabi-NH}

The analysis in subsections \ref{ax-hor-orb} and \ref{sec:NH} showed that each connected component of the horizon corresponds to an isolated point in the $x_{1}x_{2}$-plane (Lemma \ref{lem:horizon-point}) and furthermore determined the behaviour of the K\"ahler metric $h$ near each such point (eq.~\eqref{Gdd-NH} and its inverse).   We will now further restrict to toric K\"ahler metrics of Calabi type and determine the near-horizon behaviour of the coordinates $(\rho, \eta)$ and the associated functions $F(\rho), G(\eta)$.

\begin{lemma}
\label{lem:NHcalabi}
Consider a supersymmetric toric solution that is timelike outside a smooth horizon with compact cross-sections. If the K\"ahler base is of Calabi type \eqref{h-Calabi}  then, near the horizon, Calabi type coordinates are related to Gaussian null coordinates by 
\begin{equation}\label{eq:rhoeta-GNC-2}
\rho=\frac{\ell\kappa}{4\Delta_{2}(\hat{\eta})}\lambda+O(\lambda^{2})\,,\qquad\eta=\hat{\eta}+O(\lambda)\, ,
\end{equation}
where $\Delta_2(\hat{\eta})$ is given by (\ref{Deltas}), so in particular the horizon must be at $\rho=0$. Furthermore, we can always choose Calabi coordinates such that
\begin{equation}\label{eq:FG-Aux}
F(\rho)=\rho^{2}+O(\rho^3)\,,\qquad G(\eta)=(1-\eta^{2})\Delta_{1}(\eta)\, ,
\end{equation}
where $\Delta_1(\eta)$ is  given by (\ref{Deltas}).
\end{lemma}

\begin{proof}

The first step is to identify the relation between  GNC  $(\lambda, \hat{\eta})$ and the Calabi type coordinates $(\rho,\eta)$.   This is given by inverting \eqref{eq:Calabi-moment-maps}  to obtain
\begin{equation}
\rho = \frac{ \langle x, b \rangle + \langle b, c \rangle}{\langle a, b \rangle}, \qquad \eta = - \frac{\langle x,a\rangle+\langle a,c\rangle}{\langle x,b\rangle+\langle b,c\rangle}  \; , \label{rhoetax}
\end{equation}
 together with the near-horizon behaviour  \eqref{xi-NH-expl}  of the symplectic coordinates $x_i$. On the other hand, from our near-horizon analysis the Gram matrix (with respect to the K\" ahler base metric) of the toric Killing fields $G_{ij} = O(\lambda)$ as $\lambda \to 0$ (recall it is given by \eqref{Gdd-NH}), so  \eqref{Gproj} implies 
\begin{equation}\label{orders1}
\eta=O(1)\,,\qquad\frac{F(\rho)}{\rho}=O(\lambda)\,,\qquad\rho G(\eta)=O(\lambda)\,.
\end{equation}
We will now show that these relations imply that $c_i=0$.

Thus suppose, for contradiction, that $c_i$ is not identically zero.  If $\langle b, c \rangle=0$ then $b_i$ is a multiple of $c_i$ and since $\langle a, b \rangle \neq 0$ it follows that $\langle a, c \rangle \neq 0$. Then \eqref{xi-NH-expl} and \eqref{rhoetax} imply that $\eta$ is necessarily singular at the horizon contradicting (\ref{orders1}).  Therefore we must have $\langle b, c \rangle \neq 0$ in which case \eqref{xi-NH-expl}  and \eqref{rhoetax} give
\begin{align}\label{rhoeta-fake}
\rho & =\frac{\langle b,c\rangle}{\langle a,b\rangle}+\frac{1}{\langle a,b\rangle}\Big(\frac{1-\hat{\eta}}{\ca^{2}}b_{2}-\frac{1+\hat{\eta}}{\cb^{2}}b_{1}\Big)\frac{\ell\kappa\lambda}{4\Delta_{2}(\hat{\eta})}+O(\lambda^{2})\,,\nonumber \\
\eta & =-\frac{\langle a,c\rangle}{\langle b,c\rangle}+\frac{\langle a,b\rangle}{\langle b,c\rangle^{2}}\Big(\frac{1-\hat{\eta}}{\ca^{2}}c_{2}-\frac{1+\hat{\eta}}{\cb^{2}}c_{1}\Big)\frac{\ell\kappa\lambda}{4\Delta_{2}(\hat{\eta})}+O(\lambda^{2})\, .
\end{align}
Thus the horizon $\lambda= 0$ is mapped to a single point in $(\rho, \eta)$ coordinates given by  $\rho_{0}=\langle b,c\rangle / \langle a,b\rangle \neq 0$ , $\eta_{0}=-\langle a,c\rangle / \langle b,c\rangle \neq 0$. Then \eqref{orders1} implies that $F(\rho)=O(\lambda)$ and $G(\eta)=O(\lambda)$ and therefore $F(\rho_{0})=G(\eta_{0})=0$. Now expanding  \eqref{eq:GDD-Calabi} to first order in $\lambda$ we find
\begin{align}
G_{ij}  &=\Bigg[\frac{G'(\eta_{0})b_{i}b_{j}}{\langle b,c\rangle}\Big(\frac{1-\hat{\eta}}{\ca^{2}}c_{2}-\frac{1+\hat{\eta}}{\cb^{2}}c_{1}\Big) \Bigg.  \nonumber \\  &\qquad \Bigg.  +\frac{\langle a,b\rangle^{2}F'(\rho_{0})c_{i}c_{j}}{\langle b,c\rangle^{3}}\Big(\frac{1-\hat{\eta}}{\ca^{2}}b_{2}-\frac{1+\hat{\eta}}{\cb^{2}}b_{1}\Big)\Bigg]\frac{\ell\kappa\lambda}{4\Delta_{2}(\hat{\eta})}+O(\lambda^{2})\,.
\end{align}
Comparing to the corresponding expression for the near-horizon geometry \eqref{Gdd-NH} it is easy to see that the two can never match (the factor in the square brackets above is linear in $\hat\eta$ whereas the corresponding factor  in \eqref{Gdd-NH} is a quadratic or cubic polynomial in $\hat\eta$).
In order to avoid this contradiction we therefore conclude that
\begin{equation}\label{c=0}
c_i=0\, ,
\end{equation}
as claimed.

Now, combining (\ref{rhoetax}) with the near-horizon expansion  \eqref{xi-NH-expl} gives
\begin{align}\label{eq:rhoeta-GNC}
\rho & =\frac{1}{\langle a,b\rangle}\Big(\frac{1-\hat{\eta}}{\ca^{2}}b_{2}-\frac{1+\hat{\eta}}{\cb^{2}}b_{1}\Big)\frac{\ell\kappa}{4\Delta_{2}(\hat{\eta})}\lambda+O(\lambda^{2})\,,\nonumber \\
\eta & =-\frac{\ca^{2}a_{1}(1+\hat{\eta})-\cb^{2}a_{2}(1-\hat{\eta})}{\ca^{2}b_{1}(1+\hat{\eta})-\cb^{2}b_{2}(1-\hat{\eta})}+O(\lambda)\,.
\end{align}
In particular, we deduce that $\rho=0$ corresponds to the horizon as claimed in the Lemma.    Inverting, we deduce that $\lambda$ is a smooth function of $\rho$ at the horizon.

To complete the proof of our Lemma, we need to show that there exist functions $F(\rho)$
and $G(\eta)$ such that \eqref{eq:GDD-Calabi} reproduces \eqref{Gdd-NH} at $O(\lambda)$. From \eqref{orders1} and \eqref{eq:rhoeta-GNC} we see that  $F(\rho) = O(\lambda^2)$ is a smooth function of $\lambda$ and hence of $\rho$ at the horizon. We  therefore must have
\begin{equation}
F(\rho)=F_{2}\rho^{2}+O(\lambda^{3})\,,\qquad G(\eta)=G_{0}(\hat{\eta})+O(\lambda)\,,
\end{equation}
where
\begin{equation}\label{eq:F2-G0-Aux}
F_{2}=\frac{1}{2}F''(0)\,,\qquad 
G_{0}(\hat{\eta})=G\Big(-\frac{\ca^{2}a_{1}(1+\hat{\eta})-\cb^{2}a_{2}(1-\hat{\eta})}{\ca^{2}b_{1}(1+\hat{\eta})-\cb^{2}b_{2}(1-\hat{\eta})}\Big)\,.
\end{equation}
It is now not hard to see that \eqref{Gdd-NH} and  \eqref{eq:GDD-Calabi} match at $O(\lambda)$ if and only if 
\begin{equation}\label{matching-conditions}
\frac{b_{2}}{b_{1}}=-\frac{\ca^{2}}{\cb^{2}}\,,\qquad F_{2}=\frac{2}{a_{1}\ca^{2}+a_{2}\cb^{2}}\,,\qquad G_{0}(\hat{\eta})=-\frac{(1-\hat{\eta}^{2})\Delta_{1}(\hat{\eta})}{F_{2}\ca^{2}\cb^{2}b_{1}b_{2}}\,.
\end{equation}
We can now exploit the freedom in the choice of Calabi type coordinates \eqref{ab-transform} to fix 
\begin{equation}\label{ab-fix}
a_{1}=-b_{1}=\ca^{-2}\,,\qquad a_{2}=b_{2}=\cb^{-2}\,,
\end{equation}
which also fixes\footnote{Observe that $F_{2}$ is invariant under \eqref{time-resc2} and hence this freedom still remains unfixed.}
\begin{equation}\label{F2Gcub}
F_{2}=1\,,\qquad G_{0}(\hat{\eta})=(1-\hat{\eta}^{2})\Delta_{1}(\hat{\eta})\,,
\end{equation}
and  \eqref{eq:rhoeta-GNC} simplifies to \eqref{eq:rhoeta-GNC-2}. The second equation in \eqref{eq:F2-G0-Aux} now reduces to $G_{0}(\hat{\eta})=G(\hat\eta)$ which therefore determines the function $G(\eta)$ as claimed. 
\end{proof}

The above result applies to any component of the horizon and hence shows that the assumption that the base is of toric Calabi type implies the horizon must be connected, that is, it does not allow for multi-black holes.
It also shows that the near-horizon geometry has a base of Calabi type.  This follows from Lemma \ref{lem:calabi} together with the fact that the above proof shows that there are constants and functions $a_i, b_i, c_i, \rho, \eta, F(\rho), G(\eta)$ such that the symplectic coordinates $x_i$ and Gram matrix $G_{ij}$ are given by \eqref{eq:Calabi-moment-maps} and \eqref{eq:GDD-Calabi} to leading order in $\lambda$.   It is also interesting to emphasise that the near-horizon geometry completely fixes $G(\eta)$, leaving only one function $F(\rho)$ to be determined. We will see next that the latter is fixed by the supersymmetry constraint (\ref{PDE}).

\subsection{Uniqueness of supersymmetric CCLP black hole}\label{sec:Calabi-uniqueness}

In this section we will complete the proof of  Theorem \ref{main-theorem}. This first result completely determines the K\"ahler base.

\begin{lemma}
\label{lem:base}
Consider a supersymmetric toric solution that is timelike outside a smooth (analytic if $\ca^2=\cb^2$) horizon with compact cross-sections. If the K\"ahler base is of Calabi type, then it can be written as
\begin{equation}\label{Cal-CCLP}
h=\frac{\td\rho^{2}}{\rho+s\rho^{2}}+(\rho+s\rho^{2})\sigma^{2}+\frac{\rho}{(1-\eta^{2})\Delta_{1}(\eta)}(\td\eta^{2}+\tau^{2})\,,\qquad X^{(1)}=\td(\rho\sigma)\,,
\end{equation}
where $\Delta_{1}(\eta)$ is given by \eqref{Deltas} and we have defined
\begin{align}
\sigma & :=\td\psi+\eta\td\varphi=\frac{1-\eta}{\ca^{2}}\td\phi^{1}+\frac{1+\eta}{\cb^{2}}\td\phi^{2}\,,\nonumber \\
\tau & :=-(1-\eta^{2})\Delta_{1}(\eta)\td\varphi=(1-\eta^{2})\Delta_{1}(\eta)\Big(\frac{\td\phi^{1}}{\ca^{2}}-\frac{\td\phi^{2}}{\cb^{2}}\Big)\,,
\end{align}
and $s=
4/\ell^{2}$ or $s=0$, 
where $\ca^{2}$ and $\cb^{2}$ are constants that parameterise the near-horizon geometry that satisfy \eqref{kappa-constraint}.
\end{lemma}

\begin{proof} 
Recall that the near-horizon analysis in Lemma \ref{lem:NHcalabi} completely determines the function $G(\eta)$ to be a cubic that is determined by the near-horizon parameters $\ca^{2}$ and $\cb^{2}$.   We will show that the integrability condition \eqref{PDE} for Calabi type metrics (\ref{h-Calabi}) with $G(\eta)$ given by \eqref{eq:FG-Aux} can be solved for the other function $F(\rho)$.  It turns out that the cases $\ca^{2}\neq\cb^{2}$ and $\ca^{2}=\cb^{2}$ should  be treated separately since in the latter case $G(\eta)$ is quadratic rather than cubic.

For $\ca^{2}\neq\cb^{2}$ the l.h.s. of the constraint \eqref{PDE} reduces to a quadratic polynomial in $\eta$ with coefficients that depend on $F(\rho)$ and its derivatives.  Requiring that the coefficient of $\eta^{2}$ vanishes is equivalent to
\begin{equation}
\rho^{2}F''(\rho)-4\rho F'(\rho)+6F(\rho)=0\,,
\end{equation}
which has a general solution given by  $F(\rho)=F_{2}\rho^{2}+F_{3}\rho^{3}$ where $F_{2}$ and $F_{3}$ are integration constants. With this $F(\rho)$  the constraint \eqref{PDE} then vanishes identically \footnote{In fact this is the unique solution to \eqref{PDE} for cubic $G(\eta)$.}.  Next \eqref{F2Gcub} fixes $F_{2}=1$.  The constant  $F_{3}$ under the scalings \eqref{time-resc2} transforms as  $F_{3}\to K F_{3}$ and so we can use this freedom to fix $F_3$ to a convenient value which we denote by $s$. The cases $s\neq 0$ and $s=0$ are qualitatively different and  it is convenient to fix the former case so $s= 4/\ell^{2}$ \footnote{This is a convenient choice since then $f\to1$ as $\rho\to\infty$, see \eqref{f-Calabi} below.}.  Therefore, we obtain
\begin{equation}\label{eq:F-cubic}
F(\rho)=\rho^{2} +s \rho^{3}
\end{equation}
which completes the proof for the case $\ca^{2}\neq\cb^{2}$.

We now turn to the case $\ca^{2}=\cb^{2}$ which has been previously analysed in a general study of supersymmetric solutions with $SU(2)$ symmetry~\cite{Lucietti:2021bbh}. For completeness we will  reproduce the relevant arguments in the current framework. Substituting $G(\eta)$ with $\ca^{2}=\cb^{2}$ into \eqref{PDE} results in a complicated non-linear ODE for $F(\rho)$ (all $\eta$ dependence cancels in this case). This ODE is more conveniently written in terms of the function
\begin{equation}\label{FtoFcal}
\cf(\rho) :=\ca^{-2}\rho^{-2}F(\rho)\,,
\end{equation}
in terms of which the near-horizon behaviour \eqref{eq:FG-Aux} becomes
\begin{equation}\label{bdy-cond-calF}
\cf(\rho)=\ca^{-2}+O(\rho)\,.
\end{equation}
This  ODE can  be written as
\begin{equation}\label{ODE-calF1}
\Bigg(\cf^{2}\Big(12+4\cf+9\rho^{2}\big(\rho(\rho\cf)''\big)''\Big)-4\rho^{2}\cf'^{2}(\rho\cf'-3)+12\rho\cf(\rho\cf'-\cf-2)(\rho\cf')'\Bigg)''=0\,.
\end{equation}
We can integrate twice and fix the integration constants using \eqref{bdy-cond-calF} (the quantity inside the big brackets has no terms of linear order in $\rho$) to get
\begin{equation}\label{ODE-calF2}
\cf^{2}\Big(12+4\cf+9\rho^{2}\big(\rho(\rho\cf)''\big)''\Big)-4\rho^{2}\cf'^{2}(\rho\cf'-3)+12\rho\cf(\rho\cf'-\cf-2)(\rho\cf')'=4\ca^{-4}(3+\ca^{-2})\,.
\end{equation}
We will show that the only \emph{analytic} solution to \eqref{ODE-calF2} satisfying the black hole boundary condition \eqref{bdy-cond-calF} with $0<\ca^{2}<1$, is the linear function
\begin{equation}\label{uniqueSU2}
\cf(\rho)=\ca^{-2}+\cf_{1}\rho\,.
\end{equation}
To this end, we Taylor expand
\begin{equation}
\cf(\rho)=\ca^{-2}+\sum_{n=1}^{\infty}\cf_{n}\rho^{n}\,,
\end{equation}
and we will show by induction that \eqref{ODE-calF2} implies $\cf_{n}=0$ for all $n\geq 2$.
First we note that \eqref{ODE-calF2}  at order $O(\rho^{2})$ gives
\begin{equation}
(1-\cA^2)\cf_{2}=0\,,
\end{equation}
and since $0<\ca^{2}<1$ we find $\cf_{2}=0$.  Thus now assume that  $\cf_{2},\ldots,\cf_{n-1}$ all vanish for some $n \geq 3$. Then \eqref{ODE-calF2} gives 
\begin{equation}
3 (n^2-1) \Big(  3 (n^2-4)+ 8(1-\ca^2) \Big) \cf_n \rho^n+ O(\rho^{n+1})=0\,,
\end{equation}
and since $0<\ca^{2}<1$ we must have $\cf_{n}=0$. It follows by induction that $\cf_{n}=0$ for all $n\geq 3$, so we conclude that \eqref{uniqueSU2} is the unique analytic solution to \eqref{ODE-calF1} with boundary conditions \eqref{bdy-cond-calF}. In terms of $F(\rho)$, which is given by \eqref{FtoFcal}, we deduce that $F(\rho)= \rho^2 +F_3 \rho^3$ and $F_3$ can be fixed exactly as in the $\ca^{2}\neq\cb^{2}$ case to get (\ref{eq:F-cubic}).
\end{proof}

\begin{cor}
The  norm of the supersymmetric Killing field $V$ is given by
\begin{equation}\label{f-Calabi} 
f=\frac{12}{\ell^{2}}\,\frac{\rho}{3s\rho+\Delta_{2}(\eta)}\, .
\end{equation}
In particular, $V$ is strictly timelike outside the horizon.
\end{cor}

Equation \eqref{f-Calabi}  immediately follows from the first equation in \eqref{fGp} and the second in \eqref{Cal-useful} for the K\"ahler base metric  \eqref{Cal-CCLP}. From  \eqref{eq:rhoeta-GNC-2} we see that the region exterior to the horizon $\lambda>0$ corresponds to $\rho>0$ and therefore from the above explicit expression we see that $f>0$ (recall $\Delta_2$ is defined in \eqref{Deltas} and is strictly positive).  We deduce that the K\"ahler base metric is globally defined on the region exterior to the horizon.

\begin{cor}
\label{cor:axis}
The axis set in Calabi type coordinates is given by $\eta=\pm 1$ and corresponds to the fixed points of $\partial_{\phi^1}$ and $\partial_{\phi^2}$ respectively.
\end{cor}

\begin{proof}
Components of the axis are characterised by the vanishing of integer linear combinations $v= v^i \partial_{\phi^i}$ of the Killing fields.  By the previous corollary $f>0$ on the exterior region and therefore we can write the invariant $\mathbf{g}(m_i, m_j)$ on the exterior region as in (\ref{Killinginv}).   Furthermore, $\omega$ is a smooth 1-form on this region and hence $\iota_v \omega= \omega_i v^i=0$ on any axis component defined by $v=0$.
Thus on such a  component of the axis we deduce that $G_{ij} v^j=0$ and hence the K\"ahler base Gram matrix $G_{ij}$ is not full rank (to see this contract $\mathbf{g}(m_i, m_j)$ with $v^j$ and use (\ref{Killinginv})). Now for the K\"ahler metric \eqref{Cal-CCLP} one finds that the determinant of the Gram matrix in the basis  $\partial_{\phi^i}$ is $\det G_{ij}=4\ca^{4}\cb^{4}F(\rho)G(\eta)$. Therefore any component of the axis in the exterior region corresponds to $G(\eta)=0$, that is, $\eta= \pm 1$ (recall $F(\rho)>0$ away from the horizon).  From the explicit form of the K\"ahler metric \eqref{Cal-CCLP} it also follows that $\eta=1$ corresponds to the vanishing of $\partial_{\phi^1}$ and $\eta=-1$ corresponds to the vanishing of $\partial_{\phi^2}$.
\end{proof}

The next step in the uniqueness proof is to integrate  \eqref{dom} for $\omega$  and compute the gauge field from \eqref{maxwell}.  This is given by the following result.

\begin{lemma}  \label{lem:omA}
Consider a supersymmetric solution as in Lemma \ref{lem:base}.
The remaining metric data $\omega$ is given by
\begin{equation}\label{om-Calabi}
\omega=\Big(\frac{\ell^{3}s^{2}}{8}\rho+\frac{\ell^{3}s}{8}(1-\Delta_{1}(\eta))-\frac{\ell^{3}\Delta_{3}(\eta)}{48\rho}\Big)\sigma-\frac{\ell^{3}(\ca^{2}-\cb^{2})}{16}\Big(\frac{s}{\Delta_{1}(\eta)}-\frac{1}{\rho}\Big)\tau\,,
\end{equation}
and the gauge field $A$ is given (up to gauge transformations) by
\begin{equation}\label{A-Calabi}
A=\frac{\sqrt{3}}{2}f(\td t+\omega)-\frac{\ell}{2\sqrt{3}}\Big(1+\frac{3}{2}s\rho-\frac{1}{2}(\ca^{2}+\cb^{2})\Big)\sigma+\frac{\ell\sqrt{3}}{8}\frac{\ca^{2}-\cb^{2}}{\Delta_{1}(\eta)}\tau\, .
\end{equation}
\end{lemma}
    \begin{proof} Recall the 1-form $\omega$ must take the form (\ref{om}). Therefore we need to integrate \eqref{dom} for $\omega=\omega_{\psi}\td\psi+\omega_{\varphi}\td\varphi$ where the r.h.s. is given by (\ref{fGp}) and (\ref{Gm}) for the K\"ahler base in Lemma \ref{lem:base}. In particular, this requires the ASD 2-forms (\ref{eq:X2}) and (\ref{eq:X3}), which in the Calabi type coordinates of Lemma \ref{lem:base} are 
    \begin{align}
    X^{(2)} &= - \frac{\rho \,\td \rho \wedge \td \eta}{\sqrt{F(\rho)G(\eta)}}+ \sqrt{F(\rho)G(\eta)}\, \td \psi \wedge \td \varphi  \\
    X^{(3)} &= \sqrt{\frac{F(\rho)}{G(\eta)} } (\td \psi+\eta \td \varphi)\wedge \td \eta+ \sqrt{\frac{G(\eta)}{F(\rho)}} \rho \,\td \rho \wedge \td \varphi  \; .
    \end{align}
 As mentioned below equation (\ref{om}), the form of $\omega$ forces $\lambda_2=0$, thus leaving just $\lambda_3$ as an undetermined function which must be solved for alongside $\omega_\psi, \omega_\phi$, and recall these functions are all invariant under the toric symmetry. 
 
 Explicitly, the PDE \eqref{dom} for $\omega$ reads  
\begin{eqnarray}\label{om-Cal}
\partial_{\rho}\omega_{\psi}=\frac{\ell^{3}}{48}\Big(\frac{\Delta_{3}(\eta)}{\rho^{2}}+6s^{2}\Big)\,, &  & \partial_{\rho}\omega_{\varphi}=\frac{\ell^{3}}{48}\Big((6\eta s^{2}+\eta\Delta_{3}(\eta)/\rho^{2})-\rho\sqrt{\frac{G(\eta)}{F(\rho)}}\lambda_{3}\Big)\,, \\
\partial_{\eta}\omega_{\psi}=\frac{\ell^{3}}{48}\sqrt{\frac{F(\rho)}{G(\eta)}}\lambda_{3}\,,\qquad\, &  & \partial_{\eta}\omega_{\varphi}=\frac{\ell^{3}}{48}\Big(\eta\sqrt{\frac{F(\rho)}{G(\eta)}}\lambda_{3}+6s\big(\Delta_{2}(\eta)+s\rho\big)+\frac{2\Delta_{2}(\eta)^{2}+\Delta_{3}(\eta)}{\rho}\Big)\,,  \nonumber
\end{eqnarray}
where recall $F(\rho)= \rho^2+ s \rho^3$ and $G(\eta)=(1-\eta^2) \Delta_1(\eta)$. The integrability conditions for each system of equations in \eqref{om-Cal} read respectively
\begin{align}\label{lambda3-eqs}
\partial_{\rho}\Big(\sqrt{F(\rho)G(\eta)}\lambda_{3}\Big) &=\frac{3(\ca^{2}-\cb^{2})G(\eta)}{\rho^{2}}\,, \\
\partial_{\eta}\Big(\sqrt{F(\rho)G(\eta)}\lambda_{3}\Big) &=2\Big(\frac{1}{\rho}+s\Big)\big(\Delta_{2}(\eta)^{2}+\Delta_{3}(\eta)\big)\,,  \nonumber
\end{align}
which, in turn, must have an integrability condition equivalent to \eqref{PDE}. The latter is guaranteed\footnote{One can easily check this by using the identity $3(\ca^{2}-\cb^{2})G'(\eta)=-2\big(\Delta_{2}(\eta)^{2}+\Delta_{3}(\eta)\big)$.}
to be satisfied since \eqref{Cal-CCLP} is a solution of it.  We can thus  integrate \eqref{lambda3-eqs} to find
\begin{equation}
\sqrt{F(\rho)G(\eta)}\lambda_{3}=-3(\ca^{2}-\cb^{2})\frac{F(\rho)G(\eta)}{\rho^{3}}+\lambda_{3,0}\,,
\end{equation}
where $\lambda_{3,0}$ is an integration constant. Now we can integrate \eqref{om-Cal} for $\omega$ and find
\begin{align}
\omega & =\Big(\frac{\ell^{3}s^{2}}{8}\rho+\frac{\ell^{3}s}{8}(1-\Delta_{1}(\eta))-\frac{\ell^{3}\Delta_{3}(\eta)}{48\rho}\Big)\sigma-\frac{\ell^{3}(\ca^{2}-\cb^{2})}{16}\Big(\frac{s}{\Delta_{1}(\eta)}-\frac{1}{\rho}\Big)\tau+\omega_{0}\nonumber \\
 & -\frac{\ell^{3}\lambda_{3,0}}{48}\Big\{\frac{\log(s+\rho^{-1})}{(1-\eta^{2})\Delta_{1}(\eta)}\tau+\frac{1}{\ca^{2}\cb^{2}}\Big[\log\Big(\frac{2\Delta_{1}(\eta)}{1+\eta}\Big)\td\phi^{1}-\log\Big(\frac{2\Delta_{1}(\eta)}{1-\eta}\Big)\td\phi^{2}\Big]\Big\}\,,
\end{align}
where $\omega_0=\omega_{1,0}\td\phi^{1}+\omega_{2,0}\td\phi^{2}$  and $\omega_{i,0}$ are integration constants. The above expression satisfies all the \emph{local} constraints required by supersymmetry. 

Now,  the near-horizon analysis shows that $\omega_i$ must take the form \eqref{om-NH} where the $O(1)$ terms are smooth in $\lambda$ and since $\rho$ must be a smooth function of $\lambda$ near the horizon satisfying \eqref{eq:rhoeta-GNC-2},  we deduce that $\omega_i$ must diverge as $\rho^{-1}$ with the $O(1)$ terms smooth in $\rho$. Therefore, comparing to the explicit expression above we see that in order to avoid logarithmic singularities at the horizon we must have $\lambda_{3,0}=0$. Furthermore, by Corollary \ref{cor:axis} and its proof we must have that $\omega_{1}$ (resp. $\omega_{2}$) vanishes at $\eta=1$ (resp. $\eta=-1$), which requires  $\omega_{1,0}=\omega_{2,0}=0$ thus establishing  \eqref{om-Calabi}.

The gauge field can be computed from (\ref{gf}) and the Ricci form potential \eqref{Cal-useful} which yields \eqref{A-Calabi}, where we have added a gauge transformation to ensure $A$ is smooth at the axes, i.e. we have fixed $A_1=0$ at $\eta=1$ and $A_2=0$ at $\eta=-1$.
\end{proof} 

We have now completely determined the general solution under our assumptions. In particular, the solution is given by Lemma \ref{lem:base},  \eqref{f-Calabi} and Lemma \ref{lem:omA}.   It remains to show that this solution is the supersymmetric CCLP black hole (or its near-horizon geometry): we provide a convenient form for this solution in Appendix \ref{sec:CCLP} which is parameterised by two constants $A^2, B^2$.
A simple computation shows that the solution with $s=4/\ell^{2}$ is indeed the CCLP black hole as given in Appendix \ref{sec:CCLP}, upon the coordinate change
\begin{equation}
\rho=\rcclp^{2}/4\,,\qquad\eta=\cos\vartheta\,,
\end{equation}
and the identification of parameters
\begin{equation}\label{AB-identification}
\ca^{2}=A^{2}\,,\qquad\cb^{2}=B^{2}\,. 
\end{equation}
Similarly, one can show that the $s=0$ solution corresponds to the near-horizon geometry of the CCLP black hole under the same identifications. This completes the proof of Theorem \ref{main-theorem}. 

\subsection{Geometry of the CCLP K\"ahler base}\label{sec:CCLPgeometry}

In order to highlight the potency of the symplectic potential, we will elaborate a little more on the geometry of the K\"ahler base of the CCLP black hole solution and its near-horizon geometry. We will use the form of the general black hole solution derived above in Calabi coordinates with the parameters identified with those of the CCLP solution \eqref{AB-identification}.

The symplectic coordinates \eqref{eq:Calabi-moment-maps}, with \eqref{ab-fix}, are related to the Calabi type coordinates by the transformation
\begin{equation}\label{xi-rhoeta}
x_{1}=\rho\frac{1-\eta}{A^{2}}\,,\qquad x_{2}=\rho\frac{1+\eta}{B^{2}}\,,
\end{equation}
and inverting this gives
\begin{equation}\label{inv-xi-rhoeta}
\rho=\tfrac{1}{2}(A^{2}x_{1}+B^{2}x_{2})\,,\qquad\eta=-\frac{A^{2}x_{1}-B^{2}x_{2}}{A^{2}x_{1}+B^{2}x_{2}}\,.
\end{equation}
We can now exploit \eqref{eq:Calabi-sympl-potential} and  \eqref{eq:FGofAB}  to deduce the symplectic potential of any toric K\"ahler base of Calabi type.  In our case $F(\rho)$ and $G(\eta)$ are given by \eqref{eq:F-cubic} and \eqref{eq:FG-Aux} respectively and integrating  \eqref{eq:FGofAB}  we find
\begin{align}
g_0(\rho) &=\rho \log \rho - (\rho+ s^{-1}) \log (1+ s \rho)   \\
g_1(\eta) &= \frac{1}{2 A^2} (1-\eta) \log (1-\eta) +\frac{1}{2B^2} (1+\eta) \log (1+\eta) - \frac{1}{A^2 B^2} \Delta_1(\eta) \log \Delta_1(\eta)  \; ,
\end{align}
up to linear terms in $\rho$ and $\eta$ respectively which have been chosen such that the limit $s\to 0$ agrees with the $s=0$ case (such linear terms are irrelevant since by  \eqref{eq:Calabi-sympl-potential} and \eqref{inv-xi-rhoeta} they correspond to linear terms in $x_i$). Substituting these into \eqref{eq:Calabi-sympl-potential}  and using \eqref{inv-xi-rhoeta} we find the symplectic potential in symplectic coordinates is given by
\begin{align}\label{CCLP-pot-gen}
g & =\frac{1}{2}x_{1}\log x_{1}+\frac{1}{2}x_{2}\log x_{2}-\frac{1}{2}(x_{1}+x_{2})\log(x_{1}+x_{2})+\frac{1}{2}(A^{2}x_{1}+B^{2}x_{2})\log(A^{2}x_{1}+B^{2}x_{2})\nonumber \\
 & -\frac{1}{2}(A^{2}x_{1}+B^{2}x_{2}+2/s)\Big(\log(A^{2}x_{1}+B^{2}x_{2}+2/s)+\log(s/2)\Big)\,,
\end{align}
up to linear terms in $x_i$  which we have chosen such that again the limit $s\to 0$ agrees with the $s=0$ case.  
These expressions agree (up to linear terms) with the symplectic potentials for the CCLP black hole \eqref{CCLP_pot} for $s=4/\ell^{2}$ and its near-horizon geometry \eqref{NHCCLP_pot}  for $s=0$, as they must. 

As previously mentioned, the interior of the orbit space for the CCLP solutions in symplectic coordinates is the quarter plane $x_1>0, x_2>0$,  with the boundaries $x_1=0$ and $x_2=0$ corresponding to the fixed points of $\partial_{\phi_1}$ and $\partial_{\phi_2}$ respectively, and the horizon at the origin $x_1=x_2=0$.  The form of the symplectic potential takes the canonical form $g= \tfrac{1}{2}\sum_{A} \ell_A \log \ell_A$ where $\ell_A=0$ are lines in symplectic coordinates, however, only the lines defined by the axes $x_1=0$ and $x_2=0$ correspond to the boundary of the orbit space and the remaining ones sit outside as depicted in Figure \ref{fig:CCLP}.
\begin{figure}[h!]
\centering
   \includegraphics[width=0.44\textwidth]{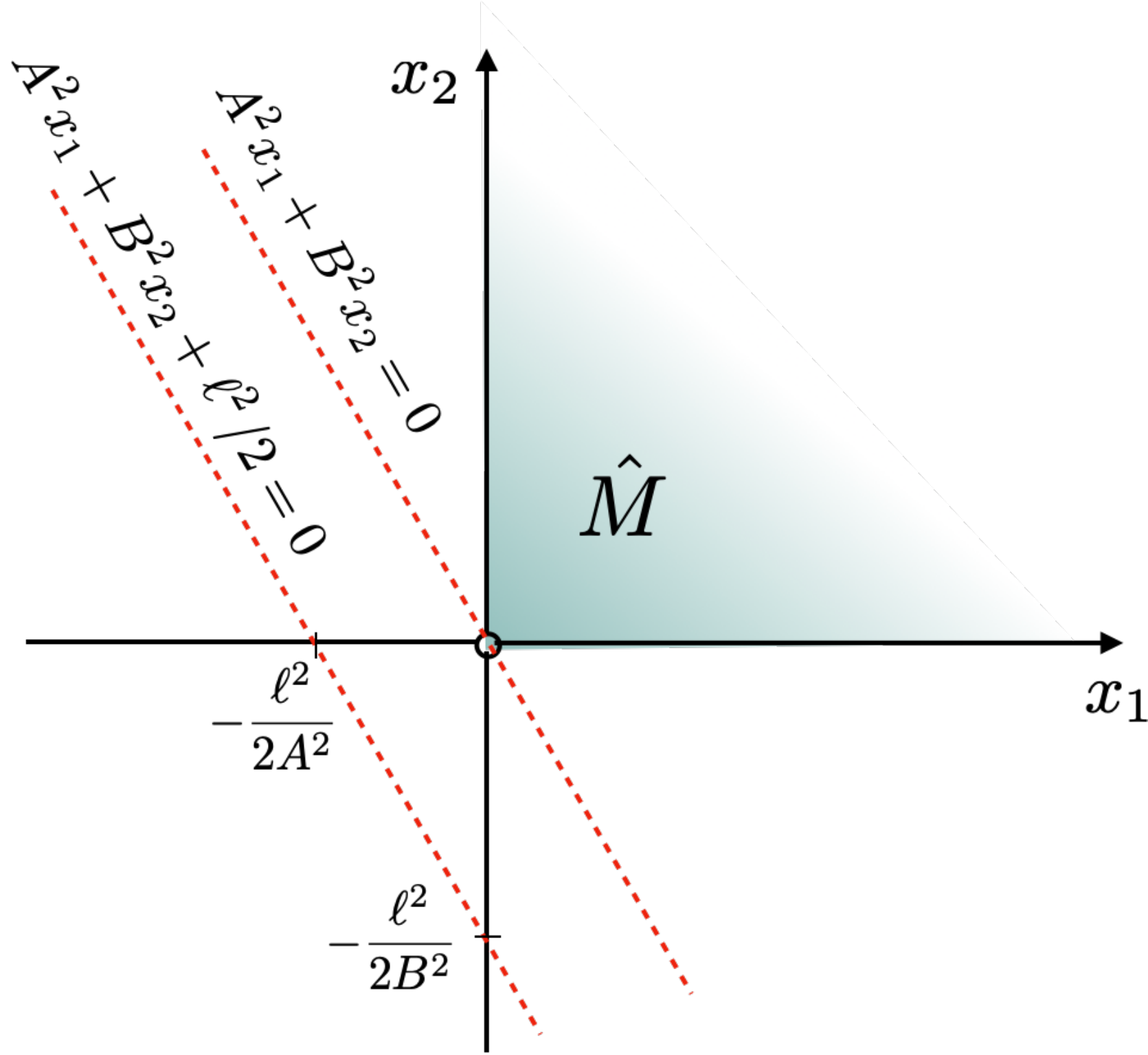}
\caption{The orbit space of  the CCLP black hole, together with the lines that define the symplectic potential.}
\label{fig:CCLP}
\end{figure}

It is instructive to see how the CCLP K\"ahler base is described in complex coordinates (see also~\cite{Figueras:2006xx}) and compare it with the symplectic formalism given above. The dictionary between the complex and symplectic formalism is detailed in Appendix \ref{app:symplcoords}. The real parts of the holomorphic coordinates $y^{i}+i \phi^{i}$ are related to the moment maps $x_{i}$ through $y^{i}=\partial^{i}g$ and from \eqref{CCLP-pot-gen} we find 
\begin{equation}\label{yi-xi}
e^{2y^{1}}=\frac{x_{1}}{x_{1}+x_{2}}\Big(\frac{2\rho}{1+s\rho}\Big)^{A^{2}}\,,\qquad e^{2y^{2}}=\frac{x_{2}}{x_{1}+x_{2}}\Big(\frac{2\rho}{1+s\rho}\Big)^{B^{2}}\,,
\end{equation}
where recall $\rho$ is given by  \eqref{inv-xi-rhoeta}. The axes $x_{1,2}=0$ correspond to $y^{1,2}\to-\infty$  with the horizon $x_{1}=x_{2}=0$ to $y^{1}=y^{2}\to-\infty$.

Let us now determine the K\"ahler potential. To this end, observe that we can rewrite  (\ref{eq:xK}) as $\partial^{i}\kpot=G^{ij}x_{j}$ using the fact that $G^{ij}= \partial y^i/\partial x_j$ (the latter follows from  (\ref{eq:yx})). Hence, given the symplectic potential in symplectic coordinates and $G^{ij}=\partial^{i}\partial^{j}g$ we can use this to determine the K\"ahler potential as a function of the symplectic coordinates.  In the case at hand we find
\begin{equation}
\partial^{1}\kpot=\frac{A^{2}}{2(1+s\rho)}\,,\qquad\partial^{2}\kpot=\frac{B^{2}}{2(1+s\rho)}\,,
\end{equation}
which readily integrates to $\kpot=s^{-1}\log(1+s\rho)$, and note that the limit $s\to 0$ agrees with the $s=0$ case which is simply  $\kpot=\rho$.  The final step is to invert \eqref{yi-xi} to express $\kpot$ as a function of the (real part of) holomorphic coordinates $y^{1}$ and $y^{2}$.  Generically, it is not possible to invert \eqref{yi-xi} analytically\footnote{For the case $A^{2}=B^{2}$, which corresponds to the GR black hole, we can do this explicitly. We have
\begin{equation}
x_{1}=\frac{e^{2y^{1}}}{A^{2}}\frac{(e^{2y_{1}}+e^{2y_{2}})^{-1}}{(e^{2y_{1}}+e^{2y_{2}})^{-1/A^{2}}-s/2}\,,\qquad x_{2}=\frac{e^{2y^{2}}}{A^{2}}\frac{(e^{2y_{1}}+e^{2y_{2}})^{-1}}{(e^{2y_{1}}+e^{2y_{2}})^{-1/A^{2}}-s/2}\,,
\end{equation}
and the K\"ahler potential reads
\begin{equation}
\kpot=-\frac{1}{s}\log\Big(1-\frac{s}{2}(e^{2y_{1}}+e^{2y_{2}})^{A^{2}}\Big)\,.
\end{equation}},
however it determines $\rho$ implicitly as a function of $y^i$. The K\"ahler potential is therefore given  by 
\begin{equation}
\kpot=\frac{1}{s}\log\big(1+s\,\rho(y^{1},y^{2})\big)\,\qquad\text{with}\qquad\Big(\frac{1+s\rho}{2\rho}\Big)^{A^{2}}e^{2y_{1}}+\Big(\frac{1+s\rho}{2\rho}\Big)^{B^{2}}e^{2y_{2}}=1\,.
\end{equation}
Equivalent expressions for the K\"ahler potential were obtained in~\cite{Figueras:2006xx}.
It is worth contrasting this {\it  implicit} expression for the K\"ahler potential to the {\it explicit} expression  for the symplectic potential \eqref{CCLP-pot-gen}.

\section{Discussion}

In this paper we have proven a uniqueness theorem for the CCLP black hole~\cite{Chong:2005hr}, which is the most general known supersymmetric AdS$_5$ black hole solution in minimal gauged supergravity. Our main assumptions are the existence of a toric symmetry that commutes with supersymmetry and that the supersymmetric solution is timelike outside a horizon, which in particular implies the geometry is determined by a toric K\"ahler metric on the base space. Furthermore, we assumed that this K\"ahler metric is of Calabi type.  The latter assumption is particularly restrictive and allows us to completely determine the general solution in terms of the near-horizon geometry data.   This result is complementary to a previous uniqueness theorem for supersymmetric black holes with $SU(2)$ symmetry~\cite{Lucietti:2021bbh}; the two classifications overlap for the special case of $SU(2)\times U(1)$ invariant solutions where one obtains a uniqueness theorem for the  Gutowski-Reall black hole (or its near-horizon geometry).

It is interesting to note that our uniqueness theorem does not make any global assumptions on the spacetime and in particular does not make any assumptions on the asymptotics of the solution. Therefore, it rules out the possibility of black holes in this symmetry class with a smooth horizon in asymptotically locally AdS$_5$ spacetimes (other than trivial global quotients of the CCLP solution).  A class of asymptotically locally AdS$_5$ black hole solutions with a squashed $S^3$ boundary and $SU(2)\times U(1)$ symmetry have been constructed numerically \cite{Blazquez-Salcedo:2017kig, Blazquez-Salcedo:2017ghg, Cassani:2018mlh}, however, these have non-smooth horizons~\cite{Lucietti:2021bbh}. In fact, since  $SU(2)\times U(1)$ invariant solutions are automatically  toric  and of Calabi type,  one can also deduce from our theorem that these numerical solutions must have non-smooth horizons.   It is interesting that for supersymmetric black holes with a smooth horizon the near-horizon geometry determines the full solution uniquely and does not allow for more general conformal boundary metrics. While it is known that supersymmetry constrains the boundary geometry~\cite{Dumitrescu:2012ha, Klare:2012gn, BenettiGenolini:2016tsn,Papadimitriou:2017kzw,Ntokos:2021duk}, our uniqueness theorem shows that for supersymmetric solutions with smooth horizons the boundary geometry is even more constrained (at least for toric Calabi type K\"aher bases).  It would be interesting to understand this phenomenon from a holographic perspective. 

The proof of our uniqueness theorem uses the classification of near-horizon geometries in minimal gauged supergravity as essential input~\cite{Kunduri:2006uh}.  In that work, a candidate near-horizon geometry for a black ring with a conical singularity on the $S^2$ factor of the horizon was found, which uplifts to a warped AdS$_3$ solution of type IIB supergravity that can be made regular in 10-dimensions. Recently, it has been shown that the horizon of this solution in five-dimensions is an orbifold known as a spindle~\cite{Ferrero:2020laf}.  It is an interesting open problem to construct an asymptotically AdS$_5$ black spindle with such a near-horizon geometry.  Since our uniqueness proof is constructive, it would be interesting to investigate whether our methods can be used to determine the existence of such black spindles and even classify them.

It is worth emphasising that our analysis concerns only supersymmetric Lorentzian signature solutions with a smooth horizon. As noted in the introduction, the supersymmetric CCLP black hole solution can be obtained from a family of non-supersymmetric non-extremal black hole solutions~\cite{Chong:2005hr} by imposing the BPS condition as well as a non-linear relation among the charges $J_{1},J_{2}$ and $Q$. Relaxing the latter results in a 3-parameter family of solutions which are \emph{locally} supersymmetric~\cite{Chong:2005hr},\cite{Ntokos:2021duk} but contain naked closed timelike curves and hence are excluded from our analysis. In fact, the K\"ahler base of such solutions belongs to the orthotoric class~\cite{Cassani:2015upa}. In Euclidean signature, we can further relax the reality of the fields and obtain a class of complex BPS saddles~\cite{Cabo-Bizet:2018ehj}, which however cannot be analytically continued to Lorentzian signature. Nonetheless, these complexified solutions are important holographically since they are dual to the dominant saddles of the field theory localisation computation (see also~\cite{Aharony:2021zkr}). An interesting extension of our analysis would be the classification of solutions of such type and to determine whether asymptotically locally AdS solutions exist in this class.

The original motivation for this work was to investigate the existence of topologically non-trivial supersymmetric black holes, such as black lenses, black holes in bubbling spacetimes and multi-black holes in AdS$_5$ (which are all known in flat space).   To accommodate such topologies one must work in the general class of solutions with toric symmetry which we investigated in the first part of this paper. Unfortunately, we found that this problem reduces to a highly nontrivial problem in toric K\"ahler geometry.   We found that symplectic coordinates appear to be best adapted to describing such solutions, in terms of which the axes of symmetry and the horizons take a simple form (line segments and points respectively). Furthermore, it is also straightforward to determine the symplectic potential (which encodes the K\"ahler metric) near any axis or horizon component. Therefore we are able to write down the singular part of this potential for possible new solutions with non-trivial topology. Unfortunately, the problem  reduces to a very complicated nonlinear 8th order PDE for this symplectic potential which prevents us from addressing the existence of such configurations.  It would be interesting to study this PDE in more detail to determine if it possesses any hidden structure.  In any case, perhaps numerical methods could be employed to construct new solutions with such non-trivial topology.

It would also be interesting to investigate the classification of supersymmetric black holes in AdS$_5$ in theories beyond minimal supergravity. Numerical evidence for  hairy supersymmetric black holes in  truncations of supergravity that retain complex scalars has been recently obtained~\cite{Markeviciute:2018yal, Markeviciute:2018cqs}. Supersymmetric timelike solutions in such truncations are found to also possess  a K\"ahler base~\cite{Liu:2007rv}. Therefore, one may hope that similar methods to those used in~\cite{Lucietti:2021bbh} and the present paper can be used to investigate the classification of supersymmetric near-horizon geometries and black holes in such theories. We leave this interesting problem for the future. \\

\noindent {\bf Acknowledgements.}   JL and PN are supported by the Leverhulme Research Project Grant RPG-2019-355.  SO is supported by a Principal's Career Development Scholarship at the University of Edinburgh.

\appendix

\section{Toric K\"ahler manifolds in symplectic coordinates}
\label{app:symplcoords}

Here we review the construction of complex and symplectic coordinates for K\"ahler toric manifolds. This is described in~\cite{Abreu}, although we fill in several details.

A K\"ahler manifold $(M, g, J)$ is toric if there is a $T^2$-action which is an isometry and is Hamiltonian. Let $K_1, K_2$ denote the two commuting Killing fields that generate the $T^2$-action (in this section $g$ is the K\"ahler metric and $J$ is the K\"ahler form or complex structure as appropriate). The Hamiltonian condition means there are globally defined moment maps $x_i$, $i=1,2$, defined by
\be
\iota_{K_i} J=- \td x_i  \; .
\ee
This implies the K\"ahler form is preserved $\mathcal{L}_{K_i} J=0$ (locally the converse is also true since $J$ is closed).   Furthermore, closedness and invariance of $J$ implies
\be
c_{ij}:= \iota_{K_i} \iota_{K_j}J  \; ,
\ee
is a constant antisymmetric matrix. Now, from the definition of the moment maps we get $\mathcal{L}_{K_i} x_j = - c_{ij}$ and therefore since we may assume the $K_i$ to have periodic flows, the $x_i$ will not be periodic functions unless $c_{ij}=0$ (if $c_{ij} \neq 0$ the function $x_i$ is monotonic  along the orbit curves of $K_{j\neq i}$). Thus we will assume 
\be
c_{ij}=0 \; ,
\ee
henceforth.

Define vector fields $X_i = J K_i$ and note that
\be
g(K_i, X_j)= g(K_i, J K_j)= g_{ab}K_i^a J^b_{~c} K^c_j= J_{ac}K^a_i K_j^c=0  \; ,
\ee
where the last equality follows from $c_{ij}=0$.  Now, since $K_i$ span 2-spaces (away from fixed points), this shows that the $X_i$ span their 2d orthogonal complements.  The invariance of $J$ together with $[K_i, K_j]=0$ implies $[K_i, X_j]= \mathcal{L}_{K_i} J K_j=0$.  Integrability of $J$ is equivalent to the vanishing of the Nijenhuis tensor
\be
0=N(X,Y)= [X,Y]+ J ( [JX,Y]+ [X,JY]) - [JX, JY]  \; ,
\ee
for all vector fields $X,Y$. In particular,  $N(K_i,K_j)=0$ then reduces to $[X_i, X_j]=0$. This shows that the 2d orthogonal complements to $\text{span}(K_1, K_2)$ are integrable distributions. 

Thus we have a commuting frame $K_i, X_j$ (recall $K_i$ and $X_i$ span 2d spaces and their orthogonal complements so must form a frame for $M$).  The dual vectors $\omega^i, J \omega^i$,\footnote{It is easy to check that the dual vectors to $X_i= JK_i$ are $J\omega^i$.} defined by $\omega^i(K_j)= \delta^i_j$ and $\omega^i(X_j)=0$, are therefore closed so we can write $\omega^i= \td t^i$ and $J\omega^i= \td y^i$ for functions $t^i, y^i$.  Thus $y^i+i t^i$ are holomorphic coordinates on $M$, i.e. $J (\td y^i+ i \td t^i)= i (\td y^i+ i \td t^i)$.  The metric in the coordinates $(y^i, t^i)$ is
\be
g = F_{ij} \td y^i \td y^j+ F_{ij} \td t^i \td t^j  \; ,
\ee
where we have used (note $K_i = \partial_{t^i}, X_j = \partial_{y^j}$)
\be
F_{ij}:= g(K_i,K_j)= g(JK_i, JK_j)= g(X_i, X_j)  \; ,
\ee
and $g_{y^i t^j} = g(X_i, K_j)=0$ from above.  

Next, using $J\td y^i = - \td t^i$ etc it then easily follows that the K\"ahler form is
\be
J= F_{ij} \td y^i \wedge \td t^j  \; ,
\ee
and closedness of $J$ is equivalent to $F_{ij} \td y^i$ being closed, which is equivalent to $F_{ij}$ being the Hessian\footnote{To prove this, note that locally there are functions $f_i$ such that $\td f_i = F_{ij} \td y^j$ which is equivalent to $f_{i,j}= F_{ij}$. Hence $f_{[i,j]}=0$ by symmetry of $F_{ij}$. Thus  locally $f_i = \partial_i f$ for some function $f$, i.e. $F_{ij}= \partial_i \partial_j f$ as required.)}
\be
F_{ij}=\partial_{y^i} \partial_{y^j} \kpot  \; ,
\ee
of some function $\kpot(y)$. It turns out that $\kpot$ is the K\"ahler potential, i.e. $J= 2i \partial \bar{\partial} \kpot$ (in complex coordinates $z^i= y^i+ i t^i$ we have $\partial_{z^i}= (1/2)( \partial_{y^i}-i \partial_{t^i})$ and note that $\kpot$ depends only on $y^i$).  Therefore, the above gives a chart $(y^i, t^i)$ for toric K\"ahler manifolds, where $y^i+i t^i$ are holomorphic coordinates, which is adapted to the  Killing fields $K_i= \partial_{t^i}$ that generate the toric action. 

There is another natural chart for toric K\"ahler manifolds which is adapted to the symplectic structure rather than the complex structure.  This is easy to deduce from the complex coordinates above. First note that computing the moment maps we find
\be
\td x_i =  F_{ij} \td y^j = \td (\partial_{y^i}  \kpot) \implies x_i = \partial_{y^i} \kpot   \; ,  \label{eq:xK}
\ee
where we have fixed an additive constant of integration for $x_i$. The inverse is given by
\be
\td y^i = G^{ij} \td x_j  \; , \label{eq:yx}
\ee
where $G^{ij}(x)$ is the inverse matrix to $F_{ij}(y)$. Integrability of this requires $G_{ij}$ to also be the Hessian 
\be
G^{ij} = \partial_{x_i} \partial_{x_j} g  \; ,
\ee
of a function $g(x)$ called the symplectic potential and we get
\be
y^i= \partial_{x_i} g \; .
\ee
Therefore, we can change coordinates to $(y^i, t^i)\to (x_i, \phi^i)$ where $\phi^i=t^i$, in terms of which the metric and K\"ahler form are
\bea
&&g= G^{ij} \td x_i \td x_j+ G_{ij} \td \phi^i \td \phi^j  \; , \\
&&J= \td x_i \wedge \td \phi^i  \; ,
\eea
where $K_i = \partial/\partial \phi^i$ and $G_{ij}$ is the inverse of $G^{ij}$. Note that the K\"ahler form is simply in Darboux coordinates (thought of as a symplectic form).  The coordinates $(x_i, \phi^i)$ are called symplectic coordinates.

\section{CCLP black hole}\label{sec:CCLP}
The supersymmetric CCLP black hole is a two parameter family of solutions first found in~\cite{Chong:2005hr}.  The K\"ahler base for this solution was determined in~\cite{Kunduri:2006ek} and we present it here essentially in the form obtained there, with a few further simplifications.

The solution depends on two parameters $A^{2},B^{2}>0$ subject to $\kappa^{2}(A^{2},B^{2})>0$ where $\kappa^{2}$ is given by \eqref{kappa-constraint}. Note that this implies $A^{2},B^{2}<1$. The K\"ahler base metric and K\"ahler form of CCLP are simply given by 
 \begin{align}\label{CCLPmetr}
h & =\frac{\td \rcclp^{2}}{V(\rcclp)}+\frac{\rcclp^{2}}{4}\Big(\frac{\td\vartheta^{2}}{\Delta_{\vartheta}}+\Delta_{\vartheta}\sin^{2}\vartheta\td\varphi^{2}\Big)+\frac{\rcclp^{2}V(\rcclp)}{4}(\td\psi+\cos\vartheta\td\varphi)^{2}\,,\nonumber \\
X^{(1)} & =\td\big(\tfrac{1}{4}\rcclp^{2}(\td\psi+\cos\vartheta\td\varphi)\big)\,,
\end{align}
where $\Delta_{\vartheta}=A^{2}\cos^{2}(\vartheta/2)+B^{2}\sin^{2}(\vartheta/2)$ and $V=1+\frac{\rcclp^2}{\ell^2}$ for CCLP and $V=1$ for its near-horizon geometry (also supersymmetric solution). In terms of $2\pi$-periodic angles $\phi^{i}$, 
\begin{equation}\label{CCLPangles}
\psi=A^{-2}\phi^{1}+B^{-2}\phi^{2}\,,\qquad\varphi=-A^{-2}\phi^{1}+B^{-2}\phi^{2}\,.
\end{equation}
The other coordinate ranges are $r\geq 0$ and $0 \leq \vartheta \leq \pi$.

The five-dimensional metric is then given by \eqref{metricform} with 
\begin{equation}
f^{-1}=1-\frac{\ell^{2}}{\rcclp^{2}}\Big(\Delta_{\vartheta}-\frac{1+A^{2}+B^{2}}{3}\Big)\,,
\end{equation}
and
\begin{align}
\omega & =\Big(\frac{\rcclp^{2}}{2\ell}+\frac{\ell}{2}(1-\Delta_{\vartheta})-\frac{\ell^{3}}{12\rcclp^{2}}\big(6\Delta_{\vartheta}-A^{4}-B^{4}+A^{2}B^{2}-2(A^{2}+B^{2})-1\big)\Big)(\td\psi+\cos\vartheta\td\varphi)\nonumber \\
 & +\frac{\ell(A^{2}-B^{2})}{4}\Big(1-\frac{\ell^{2}\Delta_{\vartheta}}{\rcclp^{2}}\Big)\sin^{2}\vartheta\td\varphi\,,
\end{align}
and the gauge field is
\begin{equation}
A=\frac{\sqrt{3}}{2}f(\td t+\omega)-\frac{\ell}{2\sqrt{3}}\Big(1+\frac{3\rcclp^{2}}{2\ell^{2}}-\frac{1}{2}(A^{2}+B^{2})\Big)(\td\psi+\cos\vartheta\td\varphi)-\frac{\ell\sqrt{3}}{8}(A^{2}-B^{2})\sin^{2}\vartheta\td\varphi\,.
\end{equation}
It is interesting to note that the CCLP solution is much simpler in this coordinate system.

For $A^2=B^2$ we recover the GR solution with $\alpha=(2A)^{-1}=(2B)^{-1}$. Then \eqref{CCLPmetr} reduces to \eqref{GRmetr}  if we also rescale  $(\psi,\varphi)_{\text{GR}}=(A^{2}\psi,A^{2}\varphi)_{\text{CCLP}}$.

In order to compare the above expressions with the ones in the literature, we give the dictionary between the notation here and in~\cite{Kunduri:2006ek}. The latter describes the generalisation of the CCLP black hole in the STU model, which admits a consistent truncation to minimal gauged supergravity with
\begin{equation}
 \frac{2}{\sqrt{3}}A_{\text{here}}=A_{\text{there}}^{1}=A_{\text{there}}^{2}=A_{\text{there}}^{3}\,.
\end{equation}
Our parameters and coordinates are related to those in~\cite[section 2.3]{Kunduri:2006ek} as follows:
\begin{equation}
\ell_{\text{here}}=g_{\text{there}}^{-1},\quad(\rcclp/\ell)_{\text{here}}=\sinh(g\scclp)_{\text{there}}\,,\quad(\phi^{1},\phi^{2},\vartheta)_{\text{here}}=(-\phi,-\psi,2\theta)_{\text{there}}\,,
\end{equation}
\begin{equation}
(A^{2},B^{2})_{\text{here}}=(A^{2},B^{2})_{\text{there}}\,,\quad\Delta_{\vartheta}|_{\text{here}}=\frac{\Delta_{\theta}}{g^{2}\alpha^{2}}\Big|_{\text{there}\,}  \; .
\end{equation}

\section{Unified form of near-horizon geometry}\label{NH-comparison}

In subsection \ref{sec:NH} we presented the general near-horizon geometry in a form which treats the solution with generic toric symmetry and enhanced $SU(2)\times U(1)$ symmetry  on the same footing (these correspond to $\ca^{2}\neq \cb^{2}$ and $\ca^{2}= \cb^{2}$ respectively).   We will show how this unified form of the near-horizon geometry is related to the original form of these near-horizon geometries which treated the two cases separately~\cite{Kunduri:2006uh, Gutowski:2004yv}.

\subsection*{Generic toric symmetry}
This near-horizon geometry was derived in~\cite{Kunduri:2006uh} and we first recall the solution in their notation. The near-horizon data, i.e. the leading order of \eqref{NH-metric}, depend on three parameters $\Gamma_{0},\Gamma_{1},\Gamma_{2}$ and explicitly read ($\chi_{\text{here}}^{i}=x_{\text{there}}^{i}$)
\begin{align}\label{NH2006}
&\Delta^{(0)}= \frac{\Delta_0}{\Gamma^2} \,, \nonumber \\
&h^{(0)}= \Gamma^{-1} \gamma_{1i} \td \chi^i  -\Gamma^{-1} \td \Gamma \,, \nonumber \\
&\gamma^{(0)}= \frac{\ell^2 \Gamma \td \Gamma^2}{4P(\Gamma)}+ \gamma_{11} \left(\td \chi^1+ \frac{\gamma_{12}}{\gamma_{11}} \td \chi^2 \right)^2+ \frac{4 P(\Gamma)}{\ell^2 \Gamma \gamma_{11}}(\td \chi^2)^2 \,,
\end{align}
where
\begin{equation}
\gamma_{11}= C^2 \Gamma-\frac{\Delta_0^2}{\Gamma^2}, 
\qquad  \gamma_{12}= \frac{\Delta_0(\alpha_0-\Gamma)}{\Gamma^2}\,,
\end{equation}
and
\begin{equation}
 P(\Gamma)= \Gamma^3- \frac{C^2\ell^2}{4}(\Gamma-\alpha_0)^2- \frac{\Delta_0^2}{C^2} 
=(\Gamma-\Gamma_{0})(\Gamma-\Gamma_{1})(\Gamma-\Gamma_{2})\,.
\end{equation}
From the latter equation we can obtain the relation between $\Delta_0, C, \alpha_0$ and $\Gamma_{0},\Gamma_{1},\Gamma_{2}$
\bea
C^2= \frac{4}{\ell^2}(\Gamma_0+\Gamma_1+\Gamma_2),  \qquad \alpha_0= \frac{\Gamma_0\Gamma_1+ \Gamma_0\Gamma_2+\Gamma_1\Gamma_2}{2(\Gamma_0+\Gamma_1+\Gamma_2)}, \nonumber  \\
\Delta_0^2=\frac{4 \Gamma_0 \Gamma_1 \Gamma_2 (\Gamma_0+\Gamma_1+\Gamma_2)}{\ell^2}-\frac{(\Gamma_0 \Gamma_1+\Gamma_0\Gamma_2+\Gamma_1 \Gamma_2)^2}{\ell^2} \; .
\eea
Observe that \eqref{NH2006} is invariant under the rescalings 
\begin{equation}\label{gamma-resc}
\Gamma\to\tilde{K}\Gamma\,,\qquad\chi^{1}\to\tilde{K}^{-1}\chi^{1}\,,\qquad\chi^{2}\to\chi^{2}\,,\qquad\Gamma_{0,1,2}\to\tilde{K}\Gamma_{0,1,2}\,,
\end{equation}
where $\tilde{K}$ is a positive constant.  

The parameters $\Gamma_{0},\Gamma_{1},\Gamma_{2}$ are constrained by
\begin{equation}\label{gamma-constraints}
0<\Gamma_{0}<\Gamma_{1}<\Gamma_{2}\,,\qquad\Delta_{0}^{2}(\Gamma_{0},\Gamma_{1},\Gamma_{2})>0\,,
\end{equation}
and the coordinate range is $\Gamma_0 \leq  \Gamma \leq \Gamma_1$  with $P(\Gamma)>0$ in the interior. At the endpoints $\Gamma= \Gamma_0, \Gamma_1$  two different linear combinations of the biaxial Killing fields vanish and the metric generically has conical singularities at these points. The biaxial  Killing fields $\partial_{\chi^{i}}$ do not necessarily have closed orbits and are related to the Killing fields with fixed points  $m_{i}=\partial_{\hat{\phi}^{i}}$  by
\be\label{m-chi}
m_1 = -d_1 \left( \frac{\gamma_{12}(\Gamma_0)}{\gamma_{11}(\Gamma_0)} \partial_{\chi^1}- \partial_{\chi^2} \right), \qquad m_2 = -d_2 \left(  \frac{\gamma_{12}(\Gamma_1)}{\gamma_{11}(\Gamma_1)}  \partial_{\chi^1}- \partial_{\chi^2} \right)\,, 
\ee
where $m_1=0$ at $\Gamma=\Gamma_0$ and $m_2=0$ at $\Gamma=\Gamma_1$. In order to avoid conical singularities $m_i$ must have closed orbits.
The constants $d_{i}$ can be determined (up to signs) by requiring that $\hat{\phi}^{i}\sim\hat{\phi}^{i}+2\pi$ and that the metric has no conical singularities at these endpoints. We find
\bea\label{di}
d_1= -\frac{\ell \left(2 \Gamma_0^2+\Gamma_0 (\Gamma_1+\Gamma_2)-\Gamma_1 \Gamma_2\right)}{2 (\Gamma_0-\Gamma_1) (\Gamma_0-\Gamma_2)} \,, \nonumber\\
d_2= \frac{\ell (\Gamma_0 (\Gamma_1-\Gamma_2)+\Gamma_1 (2 \Gamma_1+\Gamma_2))}{2 (\Gamma_0-\Gamma_1) (\Gamma_1-\Gamma_2)} \,,
\eea
where we have chosen signs to facilitate our analysis in the main text.  The transformation to the $2\pi$-periodic angles  is  given by $\chi^{i}=\hat{\phi}^{j}A_{j}^{\; i}$ where $m_{i} =A_{i}^{\; j} \partial_{\chi^j}$ and the matrix $A$ can be read off the equations \eqref{m-chi} and \eqref{di}.   

It is also worth recording that from~\cite{Kunduri:2006uh} it can be deduced that 
\be\label{Z0-2006}
Z^{(0)}= -\frac{C^2 \ell (\Gamma- \alpha_0)}{2\Gamma} \td \chi^1+\frac{2\Delta_0}{\ell C^2 \Gamma} \td \chi^2  +\frac{\ell\td \Gamma}{2 \Gamma}  \; ,
\ee
where $Z^{(0)}$ is the leading order of the 1-form $Z$ which encodes the K\"ahler form in GNC  (\ref{eq:X1GNC}).

It is useful to express the near-horizon geometry in terms of quantities invariant under \eqref{gamma-resc}. To this end, we define a new coordinate 
\begin{equation}
\hat{\eta}:=-\frac{\Gamma-\Gamma_{0}+\Gamma-\Gamma_{1}}{\Gamma_{1}-\Gamma_{0}}\,,
\end{equation}
as well as new parameters
\begin{equation}\label{AB-defs}
\ca^{2}:=\frac{\Gamma_{2}-\Gamma_{0}}{\Gamma_{0}+\Gamma_{1}+\Gamma_{2}}\,,\qquad\cb^{2}:=\frac{\Gamma_{2}-\Gamma_{1}}{\Gamma_{0}+\Gamma_{1}+\Gamma_{2}}\,.
\end{equation}
Now  the coordinate range is $-1\leq \hat{\eta}\leq 1$ and the first equation in \eqref{gamma-constraints} gives $0<\cb^{2}<\ca^{2}<1$.  It is now straightforward to verify that the expressions \eqref{NH2006} and \eqref{Z0-2006} map to \eqref{NH-expl1}, \eqref{NH-expl2},  \eqref{NH-expl3} and \eqref{Z0}, while the second constraint in \eqref{gamma-constraints} maps to \eqref{kappa2}. Some useful relations in comparing are
\begin{equation}
\Delta_{2}(\hat{\eta})=\frac{3\Gamma}{\Gamma_{0}+\Gamma_{1}+\Gamma_{2}}\,,\qquad\kappa=\frac{3\ell\Delta_{0}}{(\Gamma_{0}+\Gamma_{1}+\Gamma_{2})^{2}}\,.
\end{equation}
Finally, we can extend the range of parameters in the region $\cb^{2}>\ca^{2}$ by observing that if we exchange
\begin{equation}\label{exchange}
\hat{\phi}^{1}\leftrightarrow\hat{\phi}^{2}\,,\qquad\hat{\eta}\leftrightarrow-\hat{\eta}\,,\qquad\mathcal{A}^{2}\leftrightarrow\mathcal{B}^{2}\,,
\end{equation}
we get identical near-horizon geometries as can be easily seen from the expressions in subsection \ref{sec:NH}. 

\subsection*{Enhanced symmetry}
The near-horizon geometry with $SU(2)\times U(1)$ rotational symmetry was first derived in~\cite{Gutowski:2004yv}. In the notation of~\cite{Lucietti:2021bbh} 
the near-horizon data is parametrised by a constant $\Delta^{(0)}>\sqrt{3}/\ell$ and  given by 
\begin{align}
h^{(0)} & =-\frac{3\Delta^{(0)}}{\ell(\Delta^{(0)2}-3/\ell^{2})}\hat{\sigma}_{3}\,,\nonumber \\
\gamma^{(0)} & =\frac{1}{\Delta^{(0)2}-3/\ell^{2}}(\hat{\sigma}_{1}^{2}+\hat{\sigma}_{2}^{2})+\frac{\Delta^{(0)2}}{(\Delta^{(0)2}-3/\ell^{2})^{2}}\hat{\sigma}_{3}^{2}\,,
\end{align}
where the right-invariant 1-forms read
\begin{equation}
\hat{\sigma}_{1}=\sin\hat{\psi}\td\hat{\vartheta}-\cos\hat{\psi}\sin\hat{\vartheta}\td\hat{\varphi}\,,\qquad\hat{\sigma}_{2}=\cos\hat{\psi}\td\hat{\vartheta}+\sin\hat{\psi}\sin\hat{\vartheta}\td\hat{\varphi}\,,\qquad\hat{\sigma}_{3}=\td\hat{\psi}+\cos\hat{\vartheta}\td\hat{\varphi}\,,
\end{equation}
in terms of the Euler angles $0\leq\hat{\psi}\leq4\pi$, $0\leq\hat{\varphi}\leq2\pi$ and $0\leq\hat{\vartheta}\leq\pi$. 

The above near-horizon data corresponds precisely to the $\cb^{2}=\ca^{2}$ special case of \eqref{NH-expl1}, \eqref{NH-expl2} and \eqref{NH-expl3}. Explicitly, the two are related by the change of coordinates
\begin{equation}
\hat{\psi}=\hat{\phi}^{1}+\hat{\phi}^{2}\,,\qquad\hat{\varphi}=-\hat{\phi}^{1}+\hat{\phi}^{2}\,,\qquad
\cos\hat{\vartheta}=\hat{\eta}\,, 
\end{equation}
and change of parameters
\begin{equation}
\Delta^{(0)2}=\frac{3}{\ell^{2}}\frac{1+3\ca^{2}}{1-\ca^{2}}\,.
\end{equation}
The parameter range $\Delta^{(0)}>\sqrt{3}/\ell$ maps to $0<\ca^{2}<1$. Furthermore, $Z^{(0)}=-(\ell/3)h^{(0)}$ (see~\cite{Gutowski:2004yv}) which agrees with \eqref{Z0} with $\cb^{2}=\ca^{2}$.

\end{document}